\documentclass[journal]{IEEEtran}

\usepackage{cite}
\usepackage{amsmath, amsthm,amsfonts,mathrsfs,amssymb}
\usepackage{algorithm}
\usepackage{algorithmicx}
\usepackage{algpseudocode}
\usepackage[percent]{overpic}
\usepackage{subcaption}
\usepackage{xcolor}
\usepackage{psfrag}
\usepackage{pstricks}
\usepackage[thinlines]{easytable}

\graphicspath{{figures/}}

\newcommand{\fig}[1]{Figure{~#1}}
\newcommand{\pout}[1]{p^\mathrm{out}_{#1}}
\newcommand{\pin}[1]{p^\mathrm{in}_{#1}}
\newcommand{\tin}[1]{t^\mathrm{in}_{#1}}
\newcommand{\tout}[1]{t^\mathrm{out}_{#1}}

\newcommand{\kiter}{\xi}

\newif\iffigs
\figstrue

\newtheorem{remark}{Remark}
\newtheorem{proposition}{Proposition}
\newtheorem{assumption}{Assumption}

\definecolor{wheat}{rgb}{0.96,0.87,0.70}

\begin{document}
\title{A Semi-Distributed Interior Point Algorithm\\ for Optimal Coordination of Automated Vehicles at Intersections}
\author{Robert Hult, Mario Zanon\textsuperscript{*}, S\'{e}bastien Gros, Paolo Falcone
\thanks{\textsuperscript{*}Corresponding author. This work was partially supported by Vetenskapsr{\aa}det, grant number 2012-4038.
	
	R. Hult is with Volvo Autonomous Solutions, G\"oteborg, Sweden. e-mail: robert.hult@volvo.com.  
	
	M. Zanon is with the IMT School for Advanced Studies, Lucca Italy. e-mail: mario.zanon@imtlucca.it. 
	
	S. Gros is with the Norwegian University of Science and Technology, Trondheim, Norway. e-mail: sebastien.gros@ntnu.no. 
	
	P. Falcone is with Universit\`a degli Studi di Modena e Reggio Emilia. e-mail: falcone@unimore.it. }
}
\maketitle

\begin{abstract}
In this paper, we consider the optimal coordination of automated vehicles at intersections under fixed crossing orders.
We formulate the problem using direct optimal control and exploit the structure to construct a semi-distributed primal-dual interior-point algorithm to solve it by parallelizing most of the computations.
Differently from standard distributed optimization algorithms, where the optimization problem is split, in our approach we split the linear algebra steps, such that the algorithm takes the same steps as a fully centralized one, while still performing computations in a  distributed fashion.
We analyze the communication requirements of the algorithm, and propose an approximation scheme which can significantly reduce the data exchange.
We demonstrate the effectiveness of the algorithm in hard but realistic scenarios, which show that the approximation leads to reductions in communicated data of almost 99\% of the exact formulation, at the expense of less than 1\% suboptimality.
\end{abstract}

\IEEEpeerreviewmaketitle
\section{Introduction}
The last decade has seen a rapid development of Automated Vehicles (AV) technologies, including dedicated control, perception and communication strategies. Several standards have been adopted for vehicle-to-vehicle communication, and the use of next generation cellular communication in automotive applications is under investigation.
Consequently, the interest in applications where the AV share information and cooperate is increasing, and it is commonly held that Cooperative Automated Vehicles (CAV) will have positive effects on traffic.

One such case is the coordination of CAVs at intersections.
The idea is to let the CAVs jointly decide how to cross the intersection safely and efficiently, rather than relying on traffic-lights, road signs and traffic rules.

The literature on algorithms for coordination of CAV at intersections was surveyed in \cite{Englund2016,RiosTorres2017_kappa}, and even though most work is recent, the number of publications is growing rapidly.
While a substantial part of the literature relies completely on heuristic approaches~\cite{Dresner2008_kappa,Kowshik2011_kappa,Lee2012_kappa}, a number of contributions that employ Optimal Control (OC)~\cite{Kim2014_kappa,Katriniok2017_kappa,Britzelmeier2018_kappa,Malikopolous2018_kappa,Tallapragada2015_kappa,Riegger2016_kappa,Kneissl2018_kappa,Bali2018_kappa, Gerdts2021, katriniok2021, Mihaly2020} have been proposed recently. 
However, most OC-based algorithms partially rely on heuristics to handle the difficult combinatorial nature of the problem, which stems from the need to determine the order in which the vehicles cross the intersection.
In a number of contributions the problem is solved in two stages where 1) the crossing order is found through a heuristic, typically variations of ``first-come-first-served" \cite{Malikopolous2018_kappa, Katriniok2017_kappa, katriniok2021, Mihaly2020} or through a simplified mixed-integer optimization~\cite{Hult2018b_kappa,Hult2019c_kappa, Gerdts2021}; and 2) the control commands are found using OC-tools \cite{Riegger2016_kappa,Tallapragada2015_kappa,Kneissl2018_kappa, Gerdts2021, Hult2018b_kappa, Hult2019c_kappa}.
In this paper we propose an algorithm intended for such applications which deals with the problem of finding the optimal control commands for a fixed crossing order by relying on direct methods for OC, which transcribe the problem into a Nonlinear programming Problem (NLP). 

The fixed-order crossing problem can be solved by several approaches alternative to the one we propose in this paper, including hierarchical (e.g., bi-level) optimization, Mathematical Programming with Equilibrium Constraints (MPEC), mixed-integer NLPs (MINLP), etc. The main difficulty related to approaches based on MPEC or MINLP is the difficulty in solving these problems, which can be significantly higher than the one related to solving a Nonlinear Programming Problems (NLP). Additionally, approaches based on MPEC and MINLP are difficult to solve in a distributed fashion. Concerning hierarchical optimization, our approach can be seen as a hierarchical optimization problem in which, rather than solving the lower-level problems to full convergence, a single iterate is performed. This generally reduces the amount of computations, but can result in a slight increase in the amount of iterations taken in the upper-level problem, compared to, e.g., bi-level approaches.

In \cite{Hult2016_kappa}, we introduced a  Sequential Quadratic Programming (SQP) algorithm based on a primal decomposition of the fixed-order coordination problem, where most computations are distributed and performed on-board the vehicles in parallel.
We considered the receding horizon application of the SQP algorithm in \cite{Hult2018c_kappa}, where we also presented experimental results which demonstrated that the proposed formulation is robust with respect to packet losses, state estimation errors and unmodeled dynamics.
We extended the algorithm to handle nonlinear dynamics and economic objective functions in \cite{Hult2018a_kappa} and to handle scenarios with turning vehicles in \cite{Hult2019a_kappa}.
In \cite{Hult2018b_kappa}, we proposed an OC-based heuristic for crossing order selection, and compared the performance of our approach to  standard traffic-lights and other algorithms in \cite{Hult2019c_kappa}.
Robustness, recursive feasibility and optimality of the coordination algorithm considered in this paper have been discussed in~\cite{Hult2018c_kappa,Hult2019c_kappa}, while here we focus on distributing the computations and reducing the amount of required communication. 

\begin{figure}[t]
	\centering
	\begin{overpic}[width=0.5\linewidth]{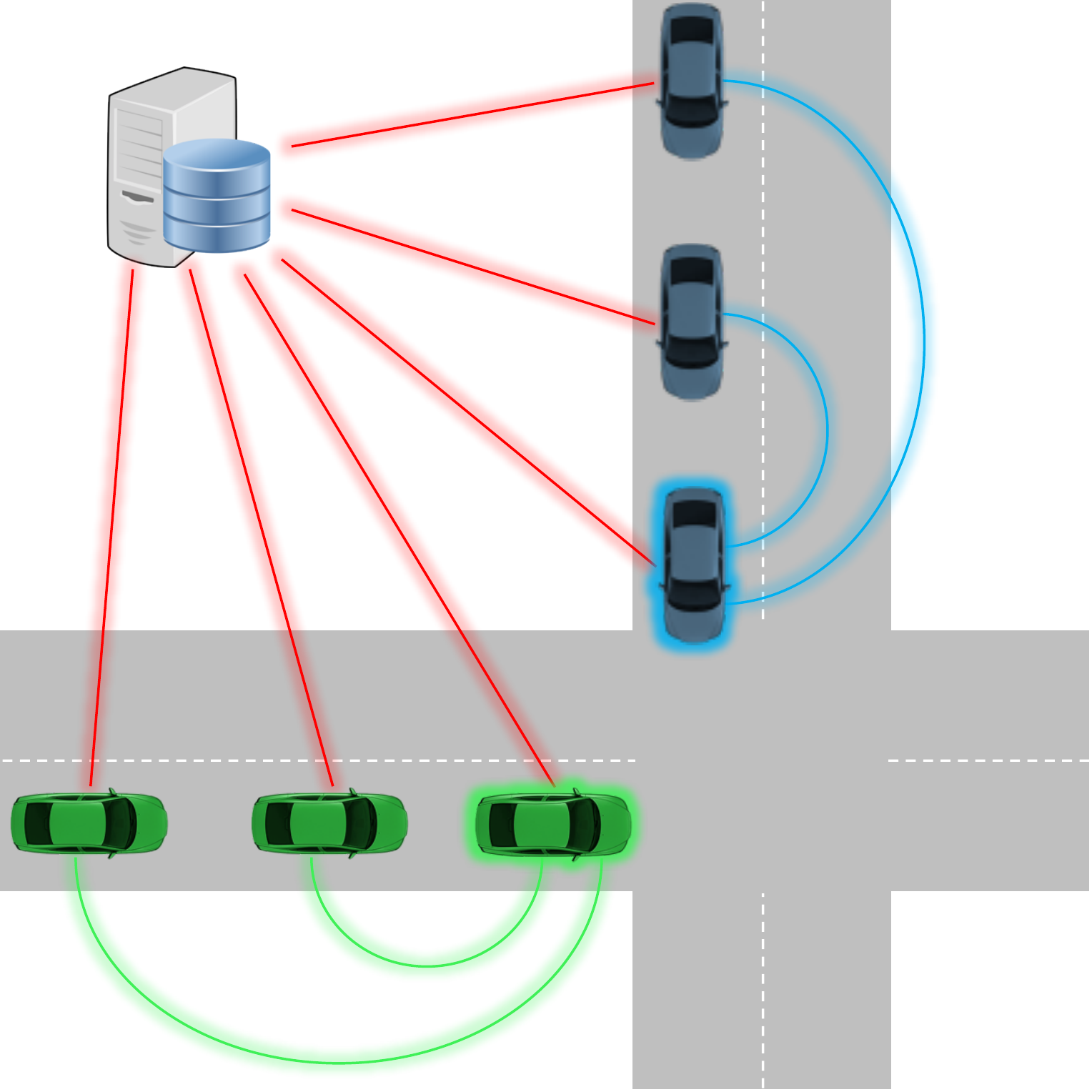}
		\put(-7,97){\footnotesize Intersection-center}
		\put(70,40){\footnotesize Lane-center}
		\put(36,15){\footnotesize Lane-center}
	\end{overpic}
	\caption{Illustration of distribution structure}\label{fig:practicalScenario}
\end{figure}

The algorithm in \cite{Hult2016_kappa} did not account for rear-end collisions between vehicles on the same lane, and required the solution of a non-smooth Nonlinear Program (NLP).
In this paper, we solve the fixed-order coordination problem by relying on Primal-Dual Interior-Point (PDIP) algorithms. Our main contribution is to make it possible to compute in a distributed manner the same steps that a centralized solver would take. 
As in \cite{Hult2016_kappa,Hult2018c_kappa}, this approach is partly centralized, and relies on central units for some computations.
In particular, the algorithm uses one intersection-wide central unit and one central unit for each lane, with communication flows as illustrated in \fig{\ref{fig:practicalScenario}}. We stress that these central units need not be physically separated from the vehicles, but a subset of the vehicles could be selected to also perform the computations of these central units.

\paragraph*{Main Contributions}
By building on ideas similar to \cite{Khosfetrat2017_kappa}, \cite{Gondzio2007}, 
we propose a way to distribute computations of PDIP schemes for nonconvex NLPs tailored to the intersection problem and our approach can be implemented in any existing PDIP solver. Furthermore, we analyze the communication requirements and propose an approach to reduce them while incurring an essentially negligible loss of optimality.

\paragraph*{Outline}
The remainder of the Paper is organized as follows.
In Section~\ref{sec:problemStatement} we model and state the intersection problem using an optimal control formalism.
In Section~\ref{sec:distributedIP} we review PDIP methods and outline how the computations can be parallelized.
In Section~\ref{sec:ap_kktSystem} we show in detail how the Karush-Kuhn-Tucker (KKT) system can be solved by splitting computations at the vehicle, lane and intersection level.
In Section~\ref{sec:ap_stepSize} we show how to select the step-size in a distributed fashion.
In Section~\ref{sec:ap_algorithm} we state a practical algorithm and provide a numerical example.
In Section~\ref{sec:ap_communication} we analyze the communication requirements and propose an approximate representation of the Rear-End Collision Avoidance (RECA) constraints, which significantly reduces the amount of data to be communicated.
The paper is concluded  in  Section~\ref{sec:ap_discussion}.

\section{Optimal Coordination at Intersections}\label{sec:problemStatement}

\begin{figure}[t]
	\centering
\begin{overpic}[width=0.5\linewidth]{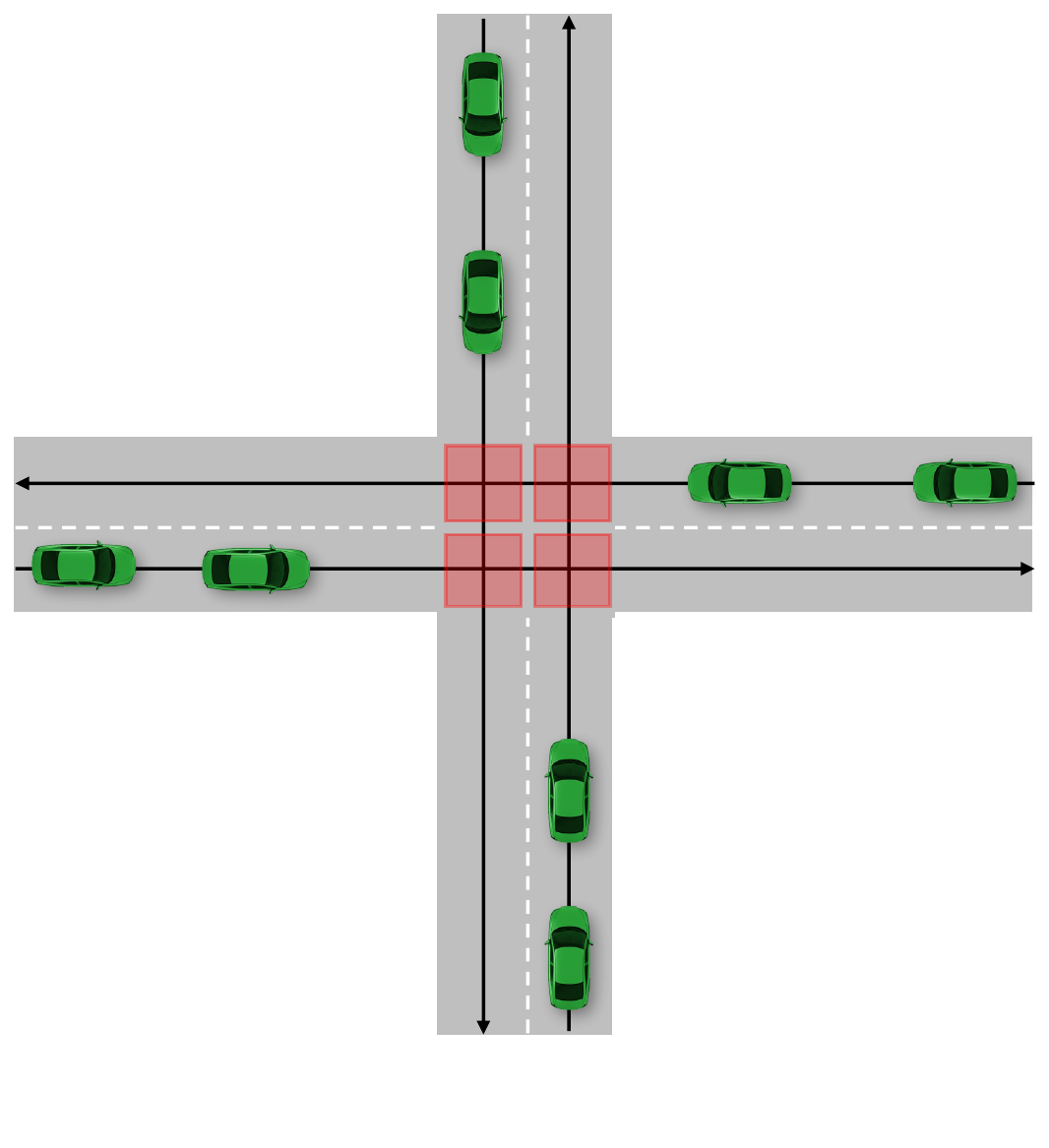}
	\put(25,64){\footnotesize CZ $1$}
	\put(54,64){\footnotesize CZ $2$}
	\put(54,41){\footnotesize CZ $3$}
	\put(25,41){\footnotesize CZ $4$}
\end{overpic}
\caption{Illustration of the scenarios considered, Assumption~\ref{ass:onRails} (black lines) and the Conflict Zones (red boxes).}\label{fig:ap_example}
\end{figure}
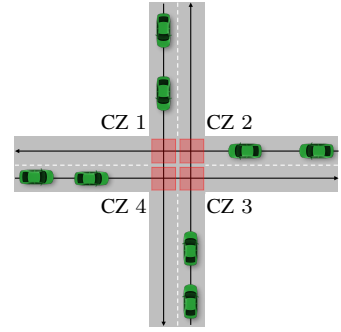

We consider intersection scenarios as shown in \fig{\ref{fig:ap_example}}, where $N$ vehicles approach an intersection with $L$ lanes, and make the following assumptions:
\begin{assumption}\label{ass:fullAuto}
There are no non-cooperative entities.
\end{assumption}
\begin{assumption}\label{ass:onRails}
The vehicles do not change lanes. 
\end{assumption}

Both assumptions are standard in the literature (see e.g. \cite{Dresner2008_kappa, Kowshik2011_kappa, Lee2012_kappa, Malikopolous2018_kappa}).
Assumption~\ref{ass:fullAuto} is introduced for the sake of simplicity and excludes the presence of, e.g., human-driven vehicles, pedestrians or bicycles. It is worth stressing that the proposed framework can be extended to accommodate for non-cooperative agents without major changes in the proposed algorithm: by modeling non-cooperative agents as uncertain systems, one can introduce additional constraints in the problem, e.g., following the approach of~\cite{Batkovic2018_kappa,Batkovic2019,Batkovic2020,Batkovic2021}.
Assumption~\ref{ass:onRails} is also introduced for simplicity and could be relaxed.
 While vehicles in general are entering and leaving the crossing area, in the following we focus on solving the problem at a given time and, for the sake of simplicity, we drop the dependence of some variables in time, e.g., the set of vehicles in each lane.
\paragraph*{Motion Models}
We describe the vehicle dynamics in continuous time as
\begin{subequations}\label{eq:dynamicsAndConstriants}
	\begin{align}
	\dot{{x}}_i(t)&=f_i({x}_i(t),{u}_i(t)),\label{eq:dynamics}\\
	0&\geq c_i({x}_i(t),{u}_i(t)) \label{eq:constraints},
	\end{align}
\end{subequations}
where $i$ is the vehicle index,  ${x}_i(t)\in\mathbb{R}^{n_i}$ and ${u}_i(t)\in\mathbb{R}^{m_i}$ are the vehicle state and control and we assume that $f_i$ is Lipschitz continuous in its first argument, such that~\eqref{eq:dynamicsAndConstriants} has a unique solution. Without loss of generality, we split the vehicle state as ${x}_i(t) = ({p}_i(t),{v}_i(t),\tilde{{x}}_i(t))$, where ${p}_i(t)$ is the position of the vehicle's geometrical center on the path describing the lane it is on, ${v}_i(t)$ is the velocity along the path and $\tilde{{x}}_i(t)$ collects (if any) all remaining states, e.g.,  acceleration and/or internal states of the powertrain. Vector $c_i$ lumps together all constraints, including, e.g., actuator limitations, lane keeping conditions, etc.
Both $f_i $ and $c_i$ are assumed to be twice differentiable.

\paragraph*{Side Collision Avoidance (SICA)}
Side collisions  can only occur between vehicles on different lanes, when these are inside a crossing area, i.e., where the lanes intersect.
We denote these areas \emph{Conflict Zones} (CZ), and note that more than one  pair $(i,j)$ can have potential collisions at a particular CZ. 
Collision avoidance consequently amounts to ensuring that vehicles on different lanes occupy each CZ  in a mutually exclusive fashion.
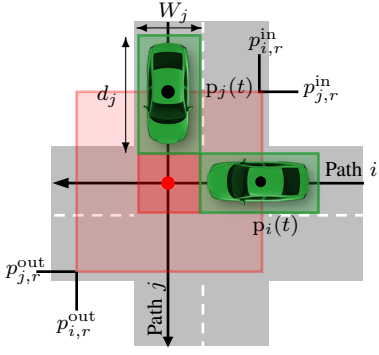
\begin{figure}[t]
	\centering
	\begin{overpic}[width=0.5\linewidth]{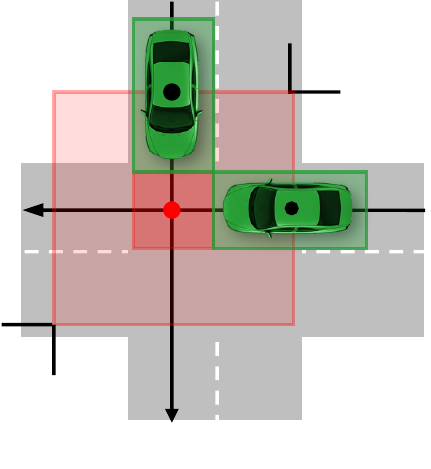}
		\put(59,93){\footnotesize $\pin{i,r}$}
		\put(-8,30){\footnotesize $\pout{j,r}$}
		\put(74,80){\footnotesize $\pin{j,r}$}
		\put(6,15){\footnotesize $\pout{i,r}$}
		\put(60,42){\footnotesize $\mathrm{p}_i(t)$}
		\put(47,80){\footnotesize $\mathrm{p}_j(t)$}
		\put(80,57){\footnotesize Path $i$}
		\put(31,13){\rotatebox{90}{\footnotesize Path $j$}}
		\put(25,80){\vector(0,1){15.5}}
		\put(25,80){\vector(0,-1){16.5}}
		\put(17,78){\footnotesize $d_j$}
		\put(37,98){\vector(1,0){9}}
		\put(37,98){\vector(-1,0){9}}
		\put(34,101){\footnotesize $W_j$}
	\end{overpic}
	\caption{Illustration of the elements used in the side collision avoidance conditions. $d_j$ and $W_j$ denotes the length and width of the vehicle, respectively. }\label{fig:sicaDef}
\end{figure}
In order to enforce this constraint, we introduce auxiliary variables for the time of entry and departure of each CZ, implicitly defined through
	\begin{align}
		\label{eq:inOutTimeDef}
		p_i(\tin{i,r}) &= \pin{i,r}, 		&
		p_i(\tout{i,r}) &= \pout{i,r},  & \forall \, r &\in \mathcal{I}^\mathrm{CZ}_i,
	\end{align}
where $\mathcal{I}^\mathrm{CZ}_i$ collects the indices of the CZ crossed by vehicle~$i$, and 
$\pin{i,r}$ and $\pout{i,r}$ are the positions along the path at which vehicle $i$ enters and leaves CZ $r$, defined as shown in \fig{\ref{fig:sicaDef}}\footnote{In the event that $\mathrm{v}_i(\tin{i,r})=0$, $\tin{i,r}$ is not  uniquely defined by $\mathrm{p}_i(\tin{i,r}) = \pin{i,r}$. In this case, one can use the slightly more involved definition $\tin{i,r} = \min t \ \mathrm{s.t.} \ \mathrm{p}_i(\tin{i,r}) = \pin{i,r}$. Since $\dot{\mathrm{p}}(\tin{i,r})=0$ would be rarely encountered in practice, this is avoided for ease of presentation.}.
SICA is then enforced as
\begin{equation}\label{eq:basicSICA}
	\tout{i,r} \leq \tin{j,r} ,  \quad (i,j,r) \in \mathcal{I}^\mathrm{C},
\end{equation}
where  $\mathcal{I}^\mathrm{C}$ collects all triples $(i,j,r)$ of vehicles $i,j$ and CZ~$r$ where side collisions can occur, and implicitly encodes the crossing order, since by~\eqref{eq:basicSICA} vehicle $i$ crosses CZ $r$ before vehicle $j$.

 \paragraph*{Rear-End Collision Avoidance (RECA)}
Due to Assumption~\ref{ass:onRails}, rear-end collisions can only occur between two adjacent vehicles on the same lane. 
We state the necessary condition for RECA as
 \begin{equation}\label{pr:eq:recaNecessary} 
 	{p}_i(t) +\delta_{i} \leq {p}_{i+1}(t), \qquad i,i+1 \in\mathcal{I}^\mathrm{v}_l.
 \end{equation}
where ${p}_i$, ${p}_{i+1}$ are the first components of the corresponding state vectors $x_i$, $x_{i+1}$ and $\delta_{i}>0$ accounts for the vehicle length and introduces a safety distance between the two vehicles and $\mathcal{I}^\mathrm{v}_{l}$ denotes the set of all vehicles on lane $l$, which we assume to be ordered, such that vehicle $i+1$ precedes vehicle $i$. 

\paragraph*{Discretization}
We employ a direct formulation of the optimal coordination problem, using a piecewise constant parameterization of the inputs ${u}_i(t) = u_{i,k}$, $t \in [t_k, t_{k+1})$, $k=1, \hdots, K-1$, where $K$ is the prediction horizon and $t_k = k\Delta t$.
We consider a multiple-shooting discretization of the dynamics \eqref{eq:dynamics}, enforcing
\begin{subequations}
	\label{eq:discreteDynamics}
	\begin{align}
		x_{i,0} & =\hat x_{i,0} \\
		x_{i,k+1}&= F_i(x_{i,k},u_{i,k},\Delta t),  & k &=0, \hdots, K-1,
	\end{align}
\end{subequations}
where $\hat x_{i,0}$ is the initial state of vehicle $i$, and $F_i(x_{i,k},u_{i,k},\Delta t)$ 
denotes the solution to  \eqref{eq:dynamics} at time $t=t_{k}+\Delta t$, with initial condition ${x}_i(t_k)=x_{i,k}$ and control $u_{i,k}$.
The state and control trajectories ${x}_i(t)$ and ${u}_i(t)$ are thereby described by ${x}_i$ and ${u}_i$, which we collect in vector ${w}_{i}=({x}_{i,0},u_{i,0}, \hdots, u_{i,K-1},x_{i,K})$. Note that the state values at time $t\neq t_k$ can be obtained from the same numerical routines used to evaluate~\eqref{eq:discreteDynamics}.
We denote the position ${p}_i(t)$ at time $t$ as the function 
\begin{equation}\label{eq:positionFunctionDiscretization}
p_i({w}_{i},t) :=F_{i,p}(x_{i,k},u_{i,k},t-t_k), \quad k = \lfloor t/\Delta t \rfloor,
\end{equation}
where $F_{i,p}$ denotes the position component of $F_i$, and we introduce $p_i({w}_{i},t)$ for ease of notation, even though this function only depends on $x_{i,k}$, $u_{i,k}$. 
Consequently, all times $\tin{i,r}$, $\tout{i,r}$ are continuous functions of $w_i$ implicitly defined as
\begin{subequations}
	\label{eq:discreteTdef}
	\begin{align}
		p_i({w}_{i},\tin{i,r}) &= \pin{i,r}, 	& \forall \, r &\in \mathcal{I}^\mathrm{CZ}_i,\\
		p_i({w}_{i},\tout{i,r}) &= \pout{i,r},  & \forall \, r&\in \mathcal{I}^\mathrm{CZ}_i.
	\end{align}
\end{subequations}

Finally, as customary in direct multiple shooting, we relax constraints \eqref{eq:constraints} and the RECA constraints~\eqref{pr:eq:recaNecessary} by enforcing them only at times $t_k$:
\begin{align}
	c_i(x_{i,k},u_{i,k}) & \leq 0,  && k = 0, \hdots, K,  \label{eq:discretePathCon}\\
	h^\mathrm{r}(p_{i},p_{i+1}) & \leq 0,  && i,i+1 \in\mathcal{I}^\mathrm{v}_l, \ l  \in \mathcal{I}^\mathrm{L}, \label{eq:discreteRECA}
\end{align}
where $h^\mathrm{r}(p_{i},p_{j}) := \left ( p_{i,k} + \delta_{i} - p_{i+1,k}, \ k \in \{0,\ldots,K \} \right )$.

\paragraph*{Optimal Control Formulation}
We define the set of all lanes $\mathcal{I}^\mathrm{L} = \{1, \hdots, L\}$; the set of all vehicles $\mathcal{I}^\mathrm{v} = \{1, \hdots, N\}=\bigcup\,  \mathcal{I}^\mathrm{v}_{l}$; and denote as $l(i)$ the lane of vehicle~$i$.
Variables $w_i$ collect the state and control trajectory: these are the only variables present when formulating an Optimal Control Problem (OCP) for a single vehicle. Since we formulate the SICA using additional time variables, we further define the in-out times $T_{i,r}=(\tin{i,r},\tout{i,r})$, relative to vehicle $i$ and CZ $r$. Then, for each vehicle we define the vector of all in-out times $T_{i}=\left (T_{i,1},\ldots,T_{i,\left |\mathcal{I}^\mathrm{CZ}_i\right |}\right ) \in \mathbb{R}^{n_{T_i}}$, which contains the in-out times for all CZ crossed by vehicle $i$. We lump all these variables together to obtain the primal variables associated with vehicle $i$ as $y_{i}=(w_{i},T_{i})$ and denote all primal variables as $y= (y_{i}, \ldots, y_{N})$. 
The optimal intersection coordination problem is then given by the NLP
\begin{subequations}\label{eq:fixedOrderProblem}
	\begin{align}
	\underset{y}{\min} & \quad \sum_{i=1}^{N}J_{i}(w_{i}) &&& \label{eq:fixedOrderProblem_objective}\\
	\mathrm{s.t.} 
	& \quad  g_{i}(w_i,T_i) = 0, &i &\in \mathcal{I}^\mathrm{v}, &&|\, \lambda_i, \label{eq:fixedOrderProblem_eqCon}\\
	& \quad  h_{i}^\mathrm{P}(w_i) \leq 0, & i &\in \mathcal{I}^\mathrm{v}, &&|\, \mu_i^\mathrm{P}, \label{eq:fixedOrderProblem_ineqCon}\\
	& \quad h^\mathrm{L}_l \left (p^\mathrm{L}_l \right ) \leq 0, &
	\ l & \in \mathcal{I}^\mathrm{L},  &&|\, \mu_{l}^\mathrm{L}
	\label{eq:fixedOrderProblem_RECAConstraints} \\
	& \quad h^\mathrm{C}(T) \leq 0, &&&&|\, \mu^\mathrm{C},  \label{eq:fixedOrderProblem_SICAConstraints}
	\end{align}
\end{subequations}
where $g_{i}(w_i,T_i)$ collects constraints \eqref{eq:discreteDynamics} and \eqref{eq:discreteTdef} for vehicle~$i$; $h_{i}^\mathrm{P}({y}_i) \leq 0$ collects path constraints \eqref{eq:discretePathCon} for vehicle~$i$; $h^\mathrm{C}(T) \leq~0$ collects SICA constraints \eqref{eq:basicSICA}; 
and
\begin{equation*}
	h^\mathrm{L}_l \left (p^\mathrm{L}_l \right ) := \left ( h^\mathrm{r}(p_i,p_{i+1}), \ \ i,i+1 \in\mathcal{I}^\mathrm{v}_l \right ),
\end{equation*} 
collects RECA constraints \eqref{eq:discreteRECA}, where we define $p^\mathrm{L}_l=(p_i, \ i \in \mathcal{I}^\mathrm{v}_l)$.
At the right of each constraint, we wrote the corresponding Lagrange multiplier.

The objective function takes the form
\begin{equation}
	J(y) = \sum_{i=1}^N J_{i}(w_{i})= \sum_{i=1}^N V_i^\mathrm{f}(x_{i,N})+\sum_{k=0}^{K-1} \ell_i(x_{i,k},u_{i,k}),
\end{equation}
with twice differentiable terminal cost $V_i^\mathrm{f}$ 
and stage cost $\ell_i$. The functions appearing in the cost are intentionally left undefined, since our approach can handle any cost function, the definition of which is problem dependent. Examples of a possible cost function include tracking costs, as in Section~\ref{sec:ap_example} and costs related to fuel consumption, as in, e.g.,~\cite{Hult2018a_kappa}. 
We observe that, since the SICA and RECA constraints are in general nonconvex,~\eqref{eq:fixedOrderProblem} is a nonconvex NLP. In the following, we assume that all functions are twice differentiable, such that derivative-based optimization algorithms can be applied. Note that this is a mild assumption in the case of CAVs.

\begin{remark}
We present Problem~\eqref{eq:fixedOrderProblem} without turning vehicles for simplicity,  but the same formalism can be deployed also in that case, as shown in \cite{Hult2019a_kappa}.
Additionally, both \eqref{eq:basicSICA} and \eqref{pr:eq:recaNecessary} can also be defined with state-dependent safety margins. 
Such details are omitted for the sake of simplicity.
\end{remark}

\begin{remark}
	In order to solve the nonconvex NLP using derivative-based approaches, we assume that all functions in~\eqref{eq:fixedOrderProblem} are continuously differentiable. Twice continuous differentiability is usually assumed for simplicity but is not strictly necessary.
\end{remark}

\section{Primal-Dual Interior Point Method}\label{sec:distributedIP}
In this paper, we focus on PDIP algorithms, which solve NLPs by relying on a smooth approximation of the KKT conditions.
By driving the smoothing (barrier) parameter to zero, a sequence of primal-dual approximations is obtained which converges to a local minimum of the NLP under mild conditions~\cite{Nocedal2006_kappa}.
Since our contribution consists in proposing a way of distributing the computations of each PDIP iteration, we briefly recall next how PDIP methods solve an NLP.

Collecting \eqref{eq:fixedOrderProblem_eqCon} in $g(y)$ and \eqref{eq:fixedOrderProblem_ineqCon}-\eqref{eq:fixedOrderProblem_SICAConstraints} in $h(y)$, the PDIP KKT conditions
of Problem \eqref{eq:fixedOrderProblem}  are
\begin{subequations}
\begin{align}
\nabla_{y}\mathcal{L} & = 0, \label{eq:pdipStat} \\
g(y) & = 0, \\
h(y)+s & = 0, \\
D(s)\mu- \mathbf{1}\tau &= 0, \label{eq:pdipComp}\\
\mu  &\geq 0,  \label{eq:pdipMu} \\
s  &\geq 0, \label{eq:pdipS}
\end{align}
\end{subequations}
where, $s$ is a slack variable associated with $h$, $\mathrm{D}(s)$ is a diagonal matrix built from $s$, $\tau \in \mathbb{R}_+$ is the barrier parameter, and $\nabla_{y} \mathcal{L}$ is the gradient of the Lagrangian function
\begin{equation}\label{eq:pdipLagrangeFunction}
	\mathcal{L}(y,\lambda,\mu,s) = J(y)+\lambda^\top g(y) + \mu^\top (h(y)+s),
\end{equation}
where $\lambda$ and $\mu$ are the Lagrange multipliers associated with constraints $g$ and $h$ respectively. 
Collecting primal-dual variables in $z=(y,\lambda,\mu,s)$, we write \eqref{eq:pdipStat}-\eqref{eq:pdipComp} as $ r_\tau(z)=0$.

Starting from an initial guess $z^{[0]}$ strictly satisfying \eqref{eq:pdipMu},~\eqref{eq:pdipS},  the sequence of primal-dual solution candidates is generated through the iteration
\begin{align}\label{eq:pdip_stepSizeIntro}
z^{[\kiter+1]} &= z^{[\kiter]} + \alpha^{[\kiter]}\Delta z^{[\kiter]},
\end{align}
where $\kiter$ is the iteration index,  $\alpha^{[\kiter]}$ the \textit{step size} and  $\Delta z^{[\kiter]}$ the \textit{search direction}, obtained as the solution of the KKT-system
\begin{equation}\label{eq:pdip_centralKKTSystemIntro}
 M\left (z^{[\kiter]}\right )\Delta z^{[\kiter]} = -  r_ {\tau^{[\kiter]}}\left (z^{[\kiter]}\right ).
\end{equation}
The KKT matrix $ M(z)$ is constructed from $\nabla_z r_\tau$, evaluated at $z^{[\kiter]}$. Typically one replaces the Lagrangian Hessian $\nabla^2_{yy}\mathcal{L}$ with an approximation $B$ to ensure that the reduced Hessian is positive-definite~\cite{Nocedal2006_kappa}.
The step size $\alpha^{[\kiter]}$ is selected such that the updated solution candidate $z^{[\kiter+1]}$ strictly satisfies  \eqref{eq:pdipMu}, \eqref{eq:pdipS} and provides sufficient decrease for a suitably selected merit function $\phi(z)$.
Finally, as the algorithm converges, an update strategy $\varphi$, enforcing $\tau^{[\kiter]}\rightarrow 0$ is used to update the barrier parameter as
\begin{equation}
	\tau^{[\kiter+1]} = \varphi\left (\tau^{[\kiter]},z^{[\kiter]}\right ).
\end{equation}

If the PDIP algorithm is applied in a fully centralized setting, the linear system \eqref{eq:pdip_centralKKTSystemIntro} is solved using a standard linear algebra routine. 
The information needed to assemble $M\left (z^{[\kiter]}\right )$ and $r_{\tau^{[\kiter]}}\left (z^{[\kiter]}\right )$ must thus be made available centrally before the search-direction $\Delta z^{[\kiter]}$ can be found.

The focus of this paper is on the solution of \eqref{eq:pdip_centralKKTSystemIntro} in a distributed fashion, by exploiting the structure of Problem \eqref{eq:fixedOrderProblem} to perform most computations independently for each vehicle and  for each lane.
Additionally, also the evaluation of the merit function can be split, allowing also the step size $\alpha^{[\kiter]}$ to be selected in a distributed fashion.
In the following sections, we detail how computations are distributed and how the information is exchanged between  the vehicles, lane centers and intersection center.

\section{Solving the PDIP-system}\label{sec:ap_kktSystem}

In this section, we construct $r_{\tau^{[\kiter]}}\left (z^{[\kiter]}\right )$ and $M\left (z^{[\kiter]}\right )$ for  Problem \eqref{eq:fixedOrderProblem} in a way which makes it possible to perform the iterates of PDIP solvers in a distributed way which is tailored to the optimal coordination of CAVs at intersections. 
Since the dimension of each slack variable $s$ is intrinsically related to a specific multiplier $\mu$, we label and index $s$ in the same way as $\mu$.
For the sake of simplicity we omit in the following the iteration index $\kiter$ and dependence on $\tau$. 

\subsection{KKT Residual $r$}
We arrange the KKT residual as $r=\left (r^\mathrm{v},r^\mathrm{L},r^\mathrm{C}\right )$ and partition the optimization variable as $z=\left (z^\mathrm{v},z^\mathrm{L},z^\mathrm{C}\right )$, where v, L, C refer to the vehicles, the lane centers and a central node. Though all components of the vector depend on each other, one can intuitively think of $r^\mathrm{v}$ as the KKT conditions relative to all vehicles;  $r^\mathrm{L}$ as the KKT conditions relative to the lanes, enforcing RECA; and $r^\mathrm{C}$ as the KKT conditions of the central node, enforcing SICA and defining the schedule. 

We define $z^\mathrm{v}=\left (z^\mathrm{v}_1,\ldots,z^\mathrm{v}_N\right )$, with $z^\mathrm{v}_i:=\left (y_i,\lambda_i,\mu_i^\mathrm{P},s_i^\mathrm{P}\right )$, and write $r^\mathrm{v}=\left (r^\mathrm{v}_1,\ldots,r^\mathrm{v}_N\right )$, with
\begin{align*}
	r^\mathrm{v}_i&:=\left (\nabla_{z_i^\mathrm{v}} \mathcal{L}, \ \mu^\mathrm{P}_{i} + \mathrm{D}\left (s_{i}^\mathrm{P}\right ) ^{-1}\mathbf{1}\tau\right ).
\end{align*}
If all vehicles were on separate lanes and there were no intersection, then all $r^\mathrm{v}_i$ could be solved independently and $\left (r^\mathrm{L},r^\mathrm{C}\right )$ would have dimension $0$. The coupling between vehicles is due to multipliers $\mu^\mathrm{L},\mu^\mathrm{C}$, which appear inside $\nabla_{z_i^\mathrm{v}} \mathcal{L}$. 
Note that $\nabla_{y_i^\mathrm{v}} \mathcal{L}=\left (\nabla_{w_i^\mathrm{v}} \mathcal{L},\nabla_{T_i^\mathrm{v}} \mathcal{L}\right )$ with %
\begin{subequations}
	\label{eq:rv}
	\begin{align} 
		\hspace{-0.3em}\nabla_{w_{i}} \mathcal{L} &= \nabla_{w_{i}} J_{i} + \nabla_{w_{i}}g_{i}\lambda_{i} + \nabla_{w_{i}}h^\mathrm{P}_{i}\mu^\mathrm{P}_{i} + \nabla_{w_i} h^\mathrm{L} \mu^\mathrm{L},  \\
		\hspace{-0.3em}\nabla_{T_{i}} \mathcal{L} &=  \nabla_{T_{i}}g_{i}\lambda_{i} + \nabla_{T_i}h^\mathrm{C}\mu^\mathrm{C},
	\end{align}
\end{subequations}
where $\nabla_{w_i} h^\mathrm{L}$ and $\nabla_{T_i}h^\mathrm{C}$ are sparse matrices composed of $1$ and $0$.

We define $z^\mathrm{L}:=\left (z^\mathrm{L}_1,\ldots,z^\mathrm{L}_L\right )$, with $z^\mathrm{L}_l:=\left (\mu^\mathrm{L}_l,s^\mathrm{L}_l\right )$, and $r^\mathrm{L}:=\left (r^\mathrm{L}_1,\ldots,r^\mathrm{L}_L\right )$ with 
\begin{align}
	\label{eq:rL}
	r^\mathrm{L}_l &:= \left (  
	h^\mathrm{L}_l\left (p^\mathrm{L}_l\right ) + s^\mathrm{L}_j
	, \ \mu^\mathrm{L}_l + \mathrm{D}\left (s^\mathrm{L}_l\right )^{-1}\mathbf{1}\tau\right ), 
\end{align}
Essentially, $r^\mathrm{L}_l$ imposes the RECA constraint on lane $l$.

Finally, we define $z^\mathrm{C}:=\left (\mu^\mathrm{C},s^\mathrm{C}\right )$, and
\begin{align}
	\label{eq:rS}
	r^\mathrm{C} := \left ( h^\mathrm{C}(T) + s^\mathrm{C}, \ \mu^\mathrm{C} + \mathrm{D}\left (s^\mathrm{C}\right )^{-1}\mathbf{1}\tau\right ),
\end{align}
imposing the SICA constraints.

\subsection{KKT Matrix $M$}

The KKT matrix is displayed schematically in Figure~\ref{fig:KKTMatrix} and an example is displayed in Figure~\ref{fig:pdip_exampleKKTMatrix}; it reads as
\begin{align}
	M=\begin{bmatrix}
	M^\mathrm{v} & M^\mathrm{vL} & M^\mathrm{vC} \\
	M^\mathrm{Lv} & M^\mathrm{L} \\
	M^\mathrm{Cv} & & M^\mathrm{C}
	\end{bmatrix},
\end{align}
where all $0$ entries are left empty for the sake of readability.

The top left block is $M^\mathrm{v}:=\mathrm{blockdiag}\left (M^\mathrm{v}_{1}, \hdots, M^\mathrm{v}_{N}\right )$, with 
\begin{multline*}
M^\mathrm{v}_{i} = 
\begin{bmatrix}
B_i & \nabla_{y_{i}} g_{i}& \nabla_{y_{i}}h_{i} & \\
\nabla_{y_{i}} {g_{i}}^\top & \\
\nabla_{y_{i}} {h_{i}}^\top & & & I \\
& & I & \mathrm{D}(s_{i})^{-1}\mathrm{D}(\mu_{i})
\end{bmatrix},
\end{multline*}
where $B_i$, $\nabla_{y_{i}} g_{i}$ and $\nabla_{y_{i}} h_{i}$ are highly sparse with the structure typical of NLPs arising in direct OC. 
Consequently, each block $M^\mathrm{v}_{i}$ can be factorized independently from the others and efficient solvers tailored to direct OC can be used~\cite{Frison2014_kappa,Domahidi2014_kappa}.

The block below $M^\mathrm{v}$ is
\begin{align*}
	M^\mathrm{Lv}:=\nabla_{z^\mathrm{v}}r^\mathrm{L}=\left (\nabla_{z^\mathrm{v}}r^\mathrm{L}_{1}, \hdots, \nabla_{z^\mathrm{v}}r^\mathrm{L}_{L}\right )=\left (M^\mathrm{Lv}_1,\ldots,M^\mathrm{Lv}_L\right ),
\end{align*}
which is block sparse, since
\begin{align*}
	M^\mathrm{Lv}_l = 
	\begin{bmatrix}
		M^\mathrm{Lv}_{l,1}, \ldots, M^\mathrm{Lv}_{l,N}
	\end{bmatrix}, && \text{with} \  M^\mathrm{Lv}_{l,i} = 0  \ \text{if} \ i \notin \mathcal{I}^\mathrm{L}_l.
\end{align*}
The fact that most entries are $0$ is best understood by noting that $M^\mathrm{Lv}$ encodes RECA constraints~\eqref{eq:discreteRECA}. Therefore, each row has only two nonzero elements equal to $\pm 1$, corresponding to position variables from $2$ vehicles on that lane, see also Equation~\eqref{eq:rL}. We will exploit this fact when analyzing the communication requirements, see Table~\ref{tab:pdip_communicationSummary}. 
Finally, $M^\mathrm{vL}:=\nabla_{z^\mathrm{L}}r^\mathrm{v}={M^\mathrm{Lv}}^\top$. We will also index this matrix by vehicle, i.e., $M^\mathrm{vL}_i := \begin{bmatrix}
M^\mathrm{vL}_{i,1}, \ldots, M^\mathrm{vL}_{i,L}
\end{bmatrix}$ only includes the rows of $M^\mathrm{vL}$ corresponding to $r^\mathrm{v}_i$.

Block $M^\mathrm{Cv}:=\nabla_{z^\mathrm{v}}r^\mathrm{C}$ is also block sparse with all nonzero elements equal to either $1$ or $-1$, since it encodes the SICA constraints~\eqref{eq:basicSICA}: we will exploit this fact when analyzing the communication requirements, see Table~\ref{tab:pdip_communicationSummary}.  The  sparsity pattern of $M^\mathrm{Cv}$ depends on the CZ crossed by each vehicle and the crossing order. In the simple case of $1$ CZ, and an appropriate vehicle ordering, this matrix can be made block diagonal with the introduction of auxiliary variables, see, e.g.,~\cite{Hult2020}. Finally, $M^\mathrm{vC}:=\nabla_{z^\mathrm{C}}r^\mathrm{v}={M^\mathrm{Cv}}^\top$. We will denote $M^\mathrm{vC}_i$ the rows corresponding to vehicle $i$, i.e., to $r^\mathrm{v}_i$.

The two remaining blocks on the diagonal are $M^\mathrm{L} =\mathrm{blockdiag}\left (M^\mathrm{L}_{1}, \hdots, M^\mathrm{L}_{L}\right )$, with 
\begin{align*}
	M^\mathrm{L}_{l} &= 
	\begin{bmatrix}
	0 & I \\
	I & \mathrm{D}(s_{l})^{-1}\mathrm{D}(\mu_{l})
	\end{bmatrix}; 
\end{align*}
and
\begin{align*}
M^\mathrm{C} &= 
\begin{bmatrix}
0 & I \\
I & \mathrm{D}(s^\mathrm{C})^{-1}\mathrm{D}(\mu^\mathrm{C})
\end{bmatrix}.
\end{align*}

\subsection{Solving the KKT-system}

\begin{figure}[t]
	\centering
	\begin{overpic}[width=0.45\linewidth]{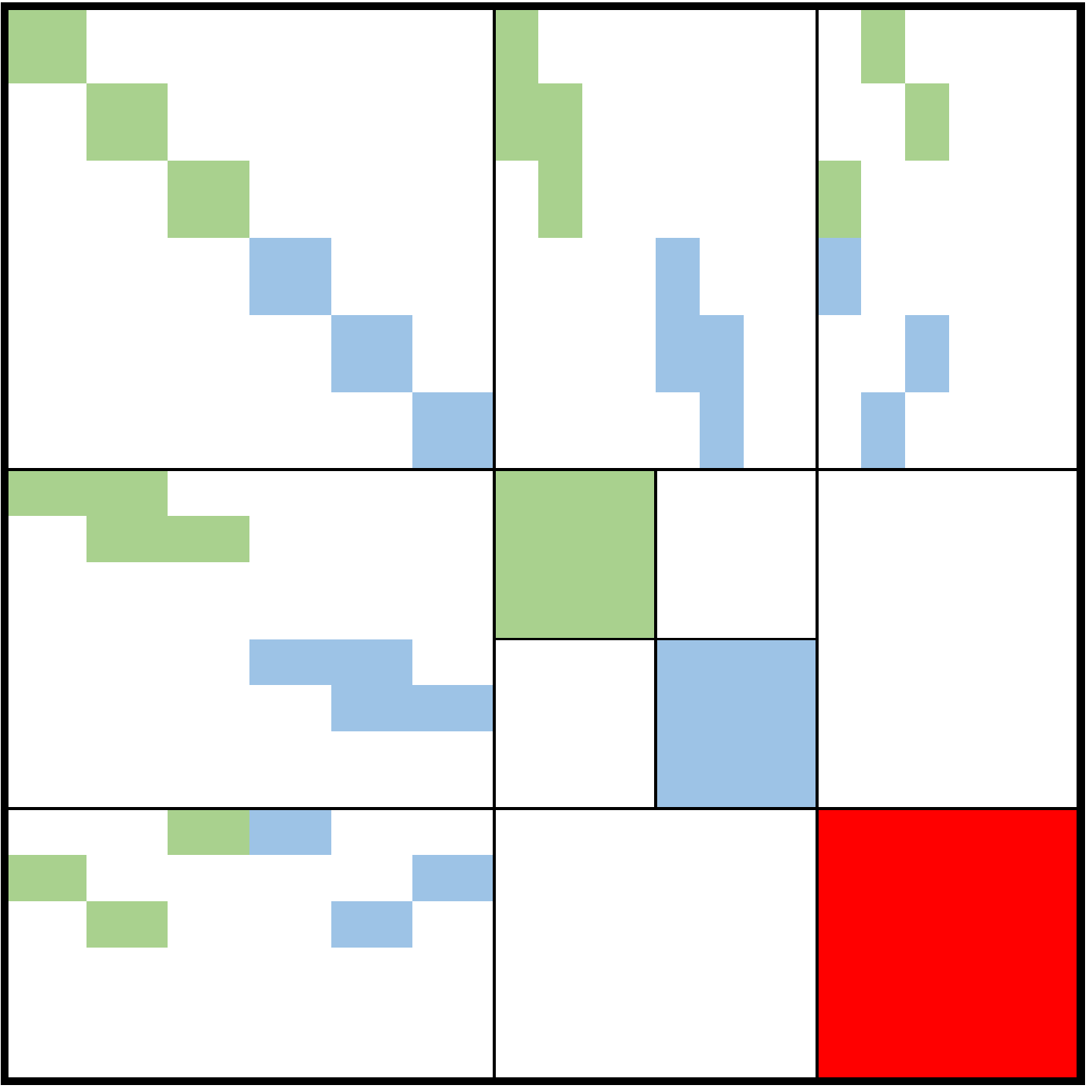}
		\put(14.5,102){\footnotesize Vehicles}
		\put(-6,68){\rotatebox{90}{\footnotesize Vehicles}}
		\put(-6,32.5){\rotatebox{90}{\footnotesize RECA}}
		\put(-6,6){\rotatebox{90}{\footnotesize SICA}}
		\put(50,102){\footnotesize RECA}
		\put(80,102){\footnotesize SICA}
		\put(19.5,40){\footnotesize $M^\mathrm{Lv}$}
		\put(19.5,15){\footnotesize $M^\mathrm{Cv}$}
		\put(55,76){\footnotesize $M^\mathrm{vL}$}
		\put(82,76){\footnotesize $M^\mathrm{vC}$}
		\put(82,12){\footnotesize $M^\mathrm{C}$}
		\put(47.5,46){\scriptsize $M^\mathrm{L}_1$}
		\put(62.5,31){\scriptsize $M^\mathrm{L}_2$}
		\put(19.5,76){\small $M^\mathrm{v}$}
	\end{overpic}
	\caption{Illustration of the KKT-matrix of \eqref{eq:fixedOrderProblem}. Blue and green differentiate vehicles on different lanes.}\label{fig:KKTMatrix}
\end{figure}
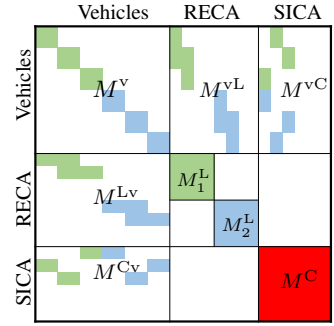

We  show next how the computations involved in solving the KKT system 
\begin{align}
	\label{eq:KKTsystem_structure}
	\begin{bmatrix}
	M^\mathrm{v} & M^\mathrm{vL} & M^\mathrm{vC} \\
	M^\mathrm{Lv} & M^\mathrm{L} \\
	M^\mathrm{Cv} & & M^\mathrm{C}
	\end{bmatrix}
	\begin{bmatrix}
	\Delta z^\mathrm{v} \\
	\Delta z^\mathrm{L} \\
	\Delta z^\mathrm{C}
	\end{bmatrix}
	=
	-
	\begin{bmatrix}
	r^\mathrm{v} \\
	r^\mathrm{L} \\
	r^\mathrm{C}
	\end{bmatrix},
\end{align}
can be split between the vehicle, lane and intersection levels of the problem. Note that this computation is common to all PDIP solvers, such that our approach can in principle be implemented in any existing solver.

\begin{proposition}
	\label{prop:lin_sys_solve}
	The KKT system \eqref{eq:KKTsystem_structure} can be solved as the following sequence of equations
	\begin{subequations}
		\label{eq:hierarchicalKKT}
		\begin{align}
		&\hspace{-2pt}\left (\hspace{-1pt}\bar M^\mathrm{C}-\bar M^\mathrm{LC}\hspace{-2pt} \left (\bar M^\mathrm{L}\right )^{-1} \hspace{-2pt} \bar M^\mathrm{CL}\hspace{-1pt}\right )\hspace{-1pt} \Delta z^\mathrm{C} = -\bar r^\mathrm{C}+\bar M^\mathrm{LC}\hspace{-2pt}\left (\bar M^\mathrm{L}\right )^{-1} \hspace{-2pt}\bar r^\mathrm{L}, \label{eq:intersectionEquations} \hspace{-10em}&&\hspace{10em} \\
		&\bar M^\mathrm{L}_l \Delta z^\mathrm{L}_l  = -\bar r^\mathrm{L}_l - \bar M^\mathrm{LC}_l\Delta z^\mathrm{C}, && \forall \, j \in \mathcal{I}^\mathrm{L} \label{eq:laneEquations}, \\
		&M^\mathrm{v}_{i} \Delta z^\mathrm{v}_{i}  =  -r^\mathrm{v}_{i} -
		\begin{bmatrix}
		M^\mathrm{vL}_i & M^\mathrm{vC}_i
		\end{bmatrix}
		\begin{bmatrix}
		\Delta z^\mathrm{L} \\
		\Delta z^\mathrm{C}
		\end{bmatrix},&& \forall i \,\in \mathcal{I}^\mathrm{v}, \label{eq:vehicleEquations}
		\end{align}
	\end{subequations}
	where we define
	\begin{subequations}
		\label{eq:Mrbar}
		\begin{align}
			\bar M^\mathrm{L} &:= \mathrm{blockdiag}\left (\bar M^\mathrm{L}_{1}, \hdots, \bar M^\mathrm{L}_{L}\right ), \label{eq:Mbar:L}\\
			\bar M^\mathrm{L}_l &:= M^\mathrm{L}_l - \sum_{i\in\mathcal{I}^\mathrm{v}_l}M^\mathrm{Lv}_{l,i} \left (M^\mathrm{v}_i\right )^{-1} M^\mathrm{vL}_{i,l}, \label{eq:Mbar:Li}\\
			\bar M^\mathrm{LC} &:= \begin{bmatrix}
			\bar M^\mathrm{LC}_1,\ldots,\bar M^\mathrm{LC}_L
			\end{bmatrix},\\
			\bar M^\mathrm{LC}_l &:= - \sum_{i\in\mathcal{I}^\mathrm{v}_l} M^\mathrm{Lv}_{l,i} \left (M^\mathrm{v}_i\right )^{-1}M^\mathrm{vC}_{i},\\
			\bar M^\mathrm{C} &:= M^\mathrm{C} - \sum_{i=1}^N M^\mathrm{Cv}_i \left (M^\mathrm{v}_i\right )^{-1}M^\mathrm{vC}_i,\label{eq:Mbar:Ci} \\
			\label{eq:rbar_0}
			\bar r^\mathrm{C} &:= r^\mathrm{C} - \sum_{i=1}^N M^\mathrm{Cv}_{i} \left (M^\mathrm{v}_i\right )^{-1} r^\mathrm{v}_i, \\
			\bar r^\mathrm{L} &:= (\bar r^\mathrm{L}_1,\ldots,\bar r^\mathrm{L}_L), \\
			\bar r^\mathrm{L}_l &:= r^\mathrm{L}_l - \sum_{i\in\mathcal{I}^\mathrm{v}_l} M^\mathrm{Lv}_{l,i}\left (M^\mathrm{v}_i\right )^{-1} r^\mathrm{v}_i,\label{eq:rbar_end}
		\end{align}
		and
		\begin{align}
			\bar M^\mathrm{LC} \left (\bar M^\mathrm{L}\right )^{-1} \bar M^\mathrm{CL} &:= \sum_{l=1}^L \bar M^\mathrm{LC}_l \left (\bar M^\mathrm{L}_l\right )^{-1} \bar M^\mathrm{CL}_l, \label{eq:Mbar:CL1} \\
			\bar M^\mathrm{LC}\left (\bar M^\mathrm{L}\right )^{-1}\bar r^\mathrm{L} &:= \sum_{l=1}^L \bar M^\mathrm{LC} \left (\bar M^\mathrm{L}_l\right )^{-1}\bar r^\mathrm{L}_l. \label{eq:Mbar:CL2}
		\end{align}
	\end{subequations}
\end{proposition}
\begin{proof}
	We proceed by first solving the KKT system with respect to $\Delta z^\mathrm{v}$, which yields
	\begin{align*}
		M^\mathrm{v} \Delta z^\mathrm{v}  =  -r^\mathrm{v} -
		\begin{bmatrix}
		M^\mathrm{vL} & M^\mathrm{vC}
		\end{bmatrix}
		\begin{bmatrix}
		\Delta z^\mathrm{L} \\
		\Delta z^\mathrm{C}
		\end{bmatrix}.
	\end{align*}
	Since $M^\mathrm{v}$ is block-diagonal, we further split this equation per vehicle to obtain~\eqref{eq:vehicleEquations}.  Since these equations require knowledge of $\Delta z^\mathrm{L}$, $\Delta z^\mathrm{C}$, they will have to be solved last.
	
	We now use~\eqref{eq:vehicleEquations} to replace 
	\begin{align}
		\label{eq:vehicle_explicit}
		\Delta z_i^\mathrm{v} = -\hbox{$M_i^\mathrm{v}$}^{^{-1}}r^\mathrm{v}_{i} - \hbox{$M_i^\mathrm{v}$}^{^{-1}}
		\begin{bmatrix}
		M^\mathrm{vL}_i & M^\mathrm{vC}_i
		\end{bmatrix}
		\begin{bmatrix}
		\Delta z^\mathrm{L} \\
		\Delta z^\mathrm{C}
		\end{bmatrix}
	\end{align}
	in the remaining equations. The first term in the right hand side of~\eqref{eq:vehicle_explicit} yields~\eqref{eq:rbar_0}-\eqref{eq:rbar_end}. The second term in the right hand side of~\eqref{eq:vehicle_explicit} yields~\eqref{eq:Mbar:L}-\eqref{eq:Mbar:Ci}. In equations~\eqref{eq:Mbar:Li} and~\eqref{eq:Mbar:Ci} the sum is restricted to $i\in\mathcal{I}_l^\mathrm{v}$ since $M_{i,l}^\mathrm{vL}=0$ for all $i\notin\mathcal{I}_l^\mathrm{v}$.
	
	We can now write~\eqref{eq:KKTsystem_structure} as
	\begin{align}
		\label{eq:KKTsystem_structure_reduced}
		\begin{bmatrix}
		\bar M^\mathrm{L\phantom{C}} & \bar M^\mathrm{LC} \\ \bar M^\mathrm{CL} & \bar M^\mathrm{C\phantom{L}}
		\end{bmatrix}
		\begin{bmatrix}
		\Delta z^\mathrm{L} \\
		\Delta z^\mathrm{C}
		\end{bmatrix}
		=
		-
		\begin{bmatrix}
		\bar r^\mathrm{L} \\
		\bar r^\mathrm{C}
		\end{bmatrix}.
	\end{align}
	We can then obtain~\eqref{eq:laneEquations} by eliminating $\Delta z^\mathrm{L}$ to obtain
	\begin{align*}
		\bar M^\mathrm{L} \Delta z^\mathrm{L}  = -\bar r^\mathrm{L} - \bar M^\mathrm{LC}\Delta z^\mathrm{C}.
	\end{align*}
	Since by~\eqref{eq:Mbar:L} $\bar M^\mathrm{L}$ is block diagonal, we can solve this equation separately for each lane, therefore obtaining~\eqref{eq:laneEquations}.
	
	Finally, Equation~\eqref{eq:intersectionEquations} is obtained by replacing
	\begin{align*}
		\Delta z^\mathrm{L}  = -\hbox{$\bar{M}^\mathrm{L}$}^{^{-1}}\bar r^\mathrm{L} - \hbox{$\bar{M}^\mathrm{L}$}^{^{-1}} \bar M^\mathrm{LC}\Delta z^\mathrm{C}
	\end{align*}
	in the second row of~\eqref{eq:KKTsystem_structure_reduced}. Note that also in~\eqref{eq:Mbar:CL1}-\eqref{eq:Mbar:CL2} we exploit the block-diagonal structure of $\bar M^\mathrm{L}$.
\end{proof}

Proposition~\ref{prop:lin_sys_solve} states that \eqref{eq:KKTsystem_structure} can be solved hierarchically, by first solving the intersection level, then the lane level and finally the vehicle level. We will discuss next where all computations are performed, since parts of the matrices and vectors are formed locally and assembled later. This is necessary both for communication and computational efficiency.

Each vehicle $i$ needs to factorize matrix $M^\mathrm{v}_i$, which has the structure typical of direct OC. This factorization can be done efficiently using, e.g.,~\cite{Frison2014_kappa,Domahidi2014_kappa}. 
Therefore, computing 
\begin{align}
	\label{eq:data_vehicle1}
	\mathcal{D}_{\mathrm{v}_i, \mathrm{L}_l} &:= \left (M^\mathrm{Lv}_{l,i}\left (M^\mathrm{v}_i\right )^{-1}M^\mathrm{vL}_{i,l}, \ M^\mathrm{Lv}_{l,i}\left (M^\mathrm{v}_i\right )^{-1} r^\mathrm{v}_i, \right .\nonumber \\
	& \hspace{104.5pt} \left. M^\mathrm{Lv}_{l,i} \left (M^\mathrm{v}_i\right )^{-1}M^\mathrm{vC}_{i}\right ), \\
	\mathcal{D}_{\mathrm{v}_i, \mathrm{C}} &:= \left (M^\mathrm{Cv}_i\left (M^\mathrm{v}_i\right )^{-1}M^\mathrm{vC}_i, \ M^\mathrm{Cv}_{i} \left (M^\mathrm{v}_i\right )^{-1} r^\mathrm{v}_i \right )	\label{eq:data_vehicle2}
\end{align}
can be done cheaply and efficiently.
Each vehicle can then communicate $\mathcal{D}_{\mathrm{v}_i, \mathrm{C}}$ to the central node and $\mathcal{D}_{\mathrm{v}_i, \mathrm{L}_{l(i)}}$ the lane center ${l(i)}$. Note that for $l\neq {l(i)}$, $\mathcal{D}_{\mathrm{v}_i, \mathrm{L}_l}$ is structurally zero. Afterwards, each lane center can compute
\begin{align}
	\label{eq:data_lane}
	\mathcal{D}_{\mathrm{L}_l, \mathrm{C}} := \left ( \bar M^\mathrm{LC}_l \left (\bar M^\mathrm{L}_l\right )^{-1} \bar M^\mathrm{CL}_l, \ \bar M^\mathrm{LC}_l \left (\bar M^\mathrm{L}_l\right )^{-1} \bar r^\mathrm{L}_l \right ).
\end{align}
Note that also these computations are relatively cheap, since the factorization of $\bar M^\mathrm{L}_l$ is needed in any case. Additionally, 
\begin{align*}
\bar M^\mathrm{L}_{l} &= 
\begin{bmatrix}
S^\mathrm{L}_l & I \\
I & \mathrm{D}(s_{l})^{-1}\mathrm{D}(\mu_{l}),
\end{bmatrix}
\end{align*}
where all blocks are diagonal except $S^\mathrm{L}_l$, which are dense. 

Once this information has been sent, the intersection center can start to solve the linear system by solving~\eqref{eq:intersectionEquations}. At this point, the intersection center needs to broadcast $\Delta z^\mathrm{C}$ to the lanes and vehicles. However, the sparsity of $\bar M^\mathrm{LC}_l$ entails that only the components of $\Delta \mu^\mathrm{S}$ relative to SICA constraints involving vehicles in lane $l$ are required, since all other components of $\Delta z^\mathrm{C}$ are multiplied by $0$: we denote these variables by $\Delta \mu^\mathrm{CL}_l$, which consist of $n_{T^\mathrm{L}_l}$ floats, where we define $n_{T^\mathrm{L}_l}:=\sum_{i\in\mathcal{I}^\mathrm{L}_l} n_{T_i}$.
The same can be stated about vehicle $i$, which only needs the components of $\Delta \mu^\mathrm{S}$ relative to SICA constraints involving vehicle $i$: we denote these variables as $\Delta \mu^\mathrm{Cv}_i$, which consist of $n_{T_i}$ floats.

Once each lane center receives $\Delta \mu^\mathrm{CL}_l$, it can solve~\eqref{eq:laneEquations}. Afterwards, each lane center needs to broadcast $\Delta z^\mathrm{L}_l$ to all vehicles on lane $l$. Similarly to the previous case, only the components of $\Delta \mu^\mathrm{L}$ relative to the RECA of vehicle $i$ are required: we denote these variables as $\Delta \mu^\mathrm{Lv}_i$. Since each vehicle has a rear and a front RECA at all times, this amounts to $2(K+1)$ floats. Once this information is available to the vehicles, they can solve~\eqref{eq:vehicleEquations}.

Algorithm~\ref{alg:ap_distributionLevel2} summarizes the procedure outlined above. Note the high degree of parallelizability: Lines~\ref{alg:ap_distributionLevel2:vehicleComp1} and \ref{alg:ap_distributionLevel2:vehicleComp2} are separable between the vehicles, and  Lines~\ref{alg:ap_distributionLevel2:laneComp1} and ~\ref{alg:ap_distributionLevel2:laneComp2} are  separable between the lanes.
Additionally, the factors of matrices $M^\mathrm{v}_{i}$, $\forall \, i \in \mathcal{I}^\mathrm{v}$ and $\bar M^\mathrm{L}_l$, $\forall \, l \in \mathcal{I}^\mathrm{L}$ computed on Lines~\ref{alg:ap_distributionLevel2:vehicleComp1} and \ref{alg:ap_distributionLevel2:laneComp1}  are stored  and reused on Lines~\ref{alg:ap_distributionLevel2:vehicleComp2} and \ref{alg:ap_distributionLevel2:laneComp2}, respectively. 

\begin{algorithm}[t]
	\caption{Distributed solution of KKT system. }
	\label{alg:ap_distributionLevel2}
	\begin{algorithmic}[1]
		\Procedure{SearchDirection}{$z$,$\tau$}
		\State $\forall \, i$: Compute $\mathcal{D}_{\mathrm{v}_i, \mathrm{C}}$,$\mathcal{D}_{\mathrm{v}_i, \mathrm{L}_{l(i)}}$, and pass to $\mathrm{C}$, $\mathrm{L}_{l(i)}$ \label{alg:ap_distributionLevel2:vehicleComp1}
		\State $\forall \, l$: Compute $\mathcal{D}_{\mathrm{L}_{l}, \mathrm{C}}$ and pass to $\mathrm{C}$ \label{alg:ap_distributionLevel2:laneComp1}
		\State $\mathrm{C}$: Solve \eqref{eq:intersectionEquations}, pass $\Delta \mu^\mathrm{CL}_{l}$, $\Delta \mu^\mathrm{Cv}_{i}$ to all $\mathrm{L}_l$, $\mathrm{v}_i$ \label{alg:ap_distributionLevel2:centralComp}
		\State $\forall \, l$: Solve \eqref{eq:laneEquations}, pass $\Delta \mu^\mathrm{Lv}_{l,i}$ to all $\mathrm{v}_i$, $i\in\mathcal{I}^\mathrm{L}_l$ 
		\label{alg:ap_distributionLevel2:laneComp2} 
		\State $\forall \, i$: Solve \eqref{eq:vehicleEquations} 
		\label{alg:ap_distributionLevel2:vehicleComp2}
		\EndProcedure
	\end{algorithmic}
\end{algorithm}

\subsection{Distributed Computation of the Step Size}\label{sec:ap_stepSize}
In this section, we discuss the selection of the step size $\alpha$ through a backtracking line search on a merit function where  most computations can be separated between, and computed in parallel within, the problem levels.

To ensure that $s^{[k+1]}>0$, $\mu^{[k+1]}>0$, we  employ the commonly used \emph{fraction from the boundary} rule, i.e., we select $\alpha\leq \alpha^\mathrm{max}$ satisfying:
	\begin{align}
	s + \alpha^\mathrm{max}\Delta s &\geq \kappa s,&
	\mu + \alpha^\mathrm{max} \Delta \mu &\geq \kappa s,
	\label{eq:pdip_fractionOfTheBoundary}
	\end{align}
where $\kappa>0$ is a parameter \cite{Nocedal2006_kappa}.
Due to the problem structure, \eqref{eq:pdip_fractionOfTheBoundary} can be evaluated separately for each vehicle, giving $\alpha_{\mathrm{v}_1}^{\mathrm{max}}$, $\forall \, i \in\mathcal{I}^\mathrm{v}$,
for the RECA constraints on a lane, giving $\alpha_{\mathrm{L}_l}^{\mathrm{max}}$, $\forall \, l \in \mathcal{I}^\mathrm{L}$ and for the SICA constraints $\alpha_{\mathrm{C}}^\mathrm{max}$.
The maximum allowed step size for the search direction $\Delta z$ is thereby
\begin{equation}\label{eq:pdip_amax}
\alpha^{\mathrm{max}} = \min\left (\alpha_{\mathrm{v}_1}^\mathrm{max}, \hdots, \alpha_{\mathrm{v}_N}^\mathrm{max}, \alpha_{\mathrm{L}_1}^\mathrm{max},\hdots,\alpha_{\mathrm{L}_L}^\mathrm{max},\alpha_{\mathrm{C}}^\mathrm{max}\right ).
\end{equation}

Additionally, for nonconvex NLPs one must ensure that $\alpha\leq \alpha^\mathrm{max}$ yields an improvement in the solution, typically 
by a backtracking line search on a suitable merit function~\cite{Nocedal2006_kappa}. In the following we consider the $\ell_1$ merit function
\begin{equation}\label{eq:pdip_meritFunction}
	\phi(y,s) = \sum_{i=1}^N\phi^\mathrm{v}_{i}\left (y_{i},s^\mathrm{P}_i\right ) + \sum_{l=1}^L \phi^\mathrm{L}_{l}\left (p,s^\mathrm{L}_{l}\right ) + \phi^\mathrm{C}\left (T,s^\mathrm{C}\right ),
\end{equation}
where
\begin{align*}
	\phi^\mathrm{v}_{i}\left (y_{i},s^\mathrm{P}_i\right ) & = J_{i}\left (w_i\right ) + \nu \left(\|g_{i}(w_i,T_i)\|_1 + \|h^\mathrm{P}_{i}(w_i)+s^\mathrm{P}_i \|_1\right) \\&\hspace{14em}- \tau \mathbf{1}^\top\log\left (s^\mathrm{P}_i\right ), \\
	\phi^\mathrm{L}_{l}\left (p,s^\mathrm{L}_{l}\right ) &= \nu\|h^\mathrm{L}_{l}\left (p^\mathrm{L}_l\right )+s^\mathrm{L}_l \|_1 - \tau\mathbf{1}^\top\log\left (s^\mathrm{L}_l\right ), \\
	\phi^\mathrm{C}\left (T,s^\mathrm{C}\right ) &= \nu\|h^\mathrm{C}(T)+s^\mathrm{C} \|_1 - \tau\mathbf{1}^\top\log\left (s^\mathrm{C}\right ), 
\end{align*}
with $\nu\geq\|(\lambda,\mu)\|_\infty$ and the logarithm taken elementwise.

We use the Armijo condition to accept a step $\alpha$ when
\begin{equation}\label{eq:pdip_armijo}
	\phi(y+\alpha\Delta y, s+\alpha\Delta s) \leq \phi(y,s) + \alpha \gamma \mathrm{D}\phi(y,s)[\Delta y,\Delta s],
\end{equation}
where $\gamma\in]0,0.5]$ is a parameter, and $\mathrm{D}\phi(y,s)[\Delta y,\Delta s]$ is the directional derivative of $\phi$ multiplied by step $(\Delta y,\Delta s)$~\cite{Nocedal2006_kappa}.

Evaluation of  $\phi$, $\mathrm{D}\phi$ can be separated between the vehicles, lanes and intersection.
We define $\Phi(\alpha):=\phi(y+\alpha\Delta y, s+\alpha\Delta s) $ ($\Phi^\mathrm{v}_i$, $\Phi^\mathrm{L}_l$, $\Phi^\mathrm{C}$ are defined equivalently) and summarize the procedure in Algorithm~\ref{alg:ap_stepSize}. 
\begin{algorithm}[t]
	\caption{Distributed selection of step-size $\alpha$, first level. }\label{alg:ap_stepSize}
	\begin{algorithmic}[1]
		\Procedure{StepSizeSelection}{$z$,$\Delta z$,$\tau$}
		\State $\forall \, i$: Compute and pass $\alpha^\mathrm{max}_{\mathrm{v}_i}$, $\Phi^\mathrm{v}_i(0)$, $\mathrm{D}\Phi^\mathrm{v}_i(0)$, $\Delta T_i$ to $C$. Pass $\Delta p_i$ to $l(i)$.  \label{list:pdip_alphamaxVehicle}
		\State $\forall \, l$: Compute and pass $\alpha^\mathrm{max}_{\mathrm{L}_l}$, $\Phi^\mathrm{L}_l(0)$, $\mathrm{D}\Phi^\mathrm{L}_l(0)$  to $C$ \label{list:pdip_alphamaxLane}
		\State $\mathrm{C}$: Compute $\alpha^\mathrm{max}_\mathrm{C}$ and $\alpha^\mathrm{max}$ with \eqref{eq:pdip_amax}, set $\alpha = \alpha^\mathrm{max}$ \hspace{-1em}\label{list:pdip_alphamaxAssembly}
		\State $\mathrm{C}$: Compute $\Phi^\mathrm{C}(0)$, $\mathrm{D}\Phi^\mathrm{C}(0)$; assemble $\Phi(0)$, $\mathrm{D}\Phi(0)$
		\Loop{}
		\State $\mathrm{C}$: Pass $\alpha$ to all $\mathrm{v}_i$, $\mathrm{L}_l$ \label{alg:pdip_stepSize:resetAlpha}
		\State $\forall \, i$: Compute and pass  $\Phi^\mathrm{v}_i(\alpha)$ to $\mathrm{C}$\label{alg:ap_StepSize:vehicleMeritFunction}
		\State $\forall \, l$: Compute and pass  $\Phi^\mathrm{L}_l(\alpha)$  to $\mathrm{C}$\label{alg:ap_StepSize:laneMeritFunction}
		\State $\mathrm{C}$: Compute $\Phi^\mathrm{C}(\alpha)$, assemble $\Phi(\alpha)$ with \eqref{eq:pdip_meritFunction}
		\If{$\Phi(\alpha)<\Phi(0) + \alpha \gamma \mathrm{D}\Phi(0)$} 
		\State \Return $\alpha$ and accept-notice to all $\mathrm{v}_i$, $\mathrm{L}_l$
		\Else
		\State $\alpha \gets \beta \alpha$
		\EndIf
		\EndLoop
		\EndProcedure
	\end{algorithmic}
\end{algorithm}

\subsection{Hessian Regularization}\label{sec:pdip_nonconvexity}
In order to ensure that $(\Delta y, \Delta s)$ is a descent direction for the merit function, i.e., that the step converges towards the solution, one must use a Lagrangian Hessian (approximation) whose reduced Hessian is positive-definite~\cite{Nocedal2006_kappa}. If one uses the exact Hessian, a regularization of the Lagrangian Hessian is necessary every time this condition is violated. 
The regularization procedure can be nontrivial, unless some simplification is introduced.

In this paper, we propose a simplification which allows us to perform the regularization in a fully decentralized fashion, based on the following observation. Assume first that no RECA nor SICA constraints are imposed: then every vehicle is independent and can regularize its Lagrangian Hessian independently. If RECA and SICA constraints are then introduced, some of the directions which were free in the search space are now constrained. Consequently, if the reduced Hessian without RECA and SICA constraints is positive definite, it must remain positive definite also after their introduction.

Note that, if $J$ is of least-squares type, the popular Gauss-Newton Hessian approximation can be employed, which is positive-(semi)definite by construction.

\begin{algorithm}[t]
	\caption{A Basic Distributed Primal-Dual Interior Point algorithm for the fixed order intersection problem. }\label{alg:ap_distributedIP}
	\begin{algorithmic}[1]
		\Procedure{FixedOrderPDIP}{$\tau^{[0]}$}
		\State $C:$ Initialize $z^\mathrm{C}$ and send  $\mu^\mathrm{C}$ to all $\mathrm{v}_i$ and $\mathrm{L}_l$ \label{alg:ap_distributedIP:centerIG}
		\State $\forall \, l$: Initialize $z^\mathrm{L}_l$ and send  $\mu^\mathrm{L}_l$ to $\mathrm{v}_i, \ i\in\mathcal{I}^\mathrm{v}_l$ \label{alg:ap_distributedIP:laneIG}
		\State $\forall \, i$: Initialize $z^\mathrm{v}_i$ and send $T^\mathrm{v}_i$ to $\mathrm{C}$, $p^\mathrm{v}_i$ to  $\mathrm{L}_l$
		\Loop
		\State $\mathrm{C}$: Send $\tau$ to all $\mathrm{v}_i$, $\mathrm{L}_l$.
		\State $\forall \, i$: Compute $M^\mathrm{v}_i,r^\mathrm{v}_i$ , if needed regularize Hessian
		\State $\Delta z \gets$\textsc{SearchDirection}($z$, $\tau$)
		\State $\alpha\gets$\textsc{StepSizeSelection}($z$, $\Delta z$, $\tau$)
		\State $\mathrm{C}$: Update $z^\mathrm{C} \gets z^\mathrm{C}+\alpha \Delta z^\mathrm{C}$ \label{alg:ap_distributedIP:centralUpdate}
		\State $\forall \, l$: Update $z^\mathrm{L}_l \gets z^\mathrm{L}_l+\alpha\Delta z^\mathrm{L}_l$
		\State $\forall \, i$: Update $z^\mathrm{v}_i \gets z^\mathrm{v}_i +\alpha\Delta z^\mathrm{v}_i$\label{alg:ap_distributedIP:vehicleUpdate}
		\If{\textsc{Terminate}($r$, $\tau$)} \label{alg:ap_distributedIP:termination}
		\State \Return Solution found
		\Else 
		\State $\tau\gets$ \textsc{UpdateBarrierParameter}($r$)\label{alg:ap_distributedIP:update}
		\EndIf
		\EndLoop
		\EndProcedure
	\end{algorithmic}
\end{algorithm}
\subsection{A Practical Algorithm}
\label{sec:ap_algorithm}

A distributed PDIP algorithm relying on Algorithms~\ref{alg:ap_distributionLevel2} and~\ref{alg:ap_stepSize}  is summarized in Algorithm~\ref{alg:ap_distributedIP}.
Note that this algorithm gives exactly the same iterates and has the same convergence properties as a fully centralized PDIP algorithm.

\subsubsection{Termination Criterion}
The algorithm is terminated on line~\ref{alg:ap_distributedIP:termination} when
\begin{equation}\label{eq:pdip_termination}
	\left \|r_{\tau^{[k]}}\left (z^{[k+1]}\right )\right \|_\infty < \varepsilon  \quad\text{and}\quad \tau^{[k]} < \varepsilon,
\end{equation}
for some tolerance $\varepsilon$.
While termination must be decided centrally, we can exploit the fact that
\begin{align*}
	\|r\|_\infty &= \max\left (\left \|r^\mathrm{v}\right \|_\infty,\left \|r^\mathrm{L}\right \|_\infty,\left \|r^\mathrm{C}\right \|_\infty\right ),\\
	\left \|r^\mathrm{v}\right \|_\infty &= \max\left (\left \|r^\mathrm{v}_1\right \|_\infty, \ldots,\|r^\mathrm{v}_N\|_\infty\right ), \\
	\left \|r^\mathrm{L}\right \|_\infty &= \max\left (\left \|r^\mathrm{L}_1\right \|_\infty, \ldots,\left \|r^\mathrm{L}_L\right \|_\infty\right ).
\end{align*}

\subsubsection{Barrier Parameter Update}
While any update scheme can be used at line~\ref{alg:ap_distributedIP:update}, in this paper we use the simple Fiacco-McCormick rule~\cite{Nocedal2006_kappa}: $\tau^{[k+1]}\gets\eta\tau^{[k]}$, with $\eta\in]0,1[$, when $\left \|r_{\tau^{[k]}}\left (z^{[k+1]}\right )\right \|_\infty < \tau^{[k]}$.

\subsection{Example}
\label{sec:ap_example}
As an example, we consider a scenario with three vehicles on each lane. 
Because we solve the problem to full convergence, the solution is independent of the used algorithm. Consequently, we do not investigate here the robustness of our approach with respect to packet losses, state estimation errors and unmodeled dynamics which we discussed in~\cite{Hult2018c_kappa}, where we have also shown that linear dynamics can be sufficient to perform trajectory tracking, though nonlinear dynamics have been used in~\cite{Hult2018a_kappa} to optimize fuel consumption. While the difficulty of any OCP depends on many aspects, including the initial guess and nonlinearity of the dynamical model, a discussion on which model is best is beyond the scope of this paper.
Assuming that all vehicles are electric, their motion can be described by
\begin{subequations}\label{eq:pdip_dynamicEquations}
	\begin{align}
	\dot{p}_i(t)&=v_i(t), \label{eq:pdip_dynamicEquations1}\\
	\dot{v}_i(t)&= \frac{1}{m_i}\left (c^E E_i(t) - F^\mathrm{B}_{i} - c^\mathrm{d}\mathrm{v}_i(t)^2-c^\mathrm{r}\right ), \label{eq:pdip_dynamicEquations2}\\
	E(t) & \leq \min\left (E_i^\mathrm{max},P_i^\mathrm{max}/\omega_i(t)\right ) \\
	0&\leq\omega_i(t) \leq \omega_i^\mathrm{max},
	\end{align}
\end{subequations}
where $E_i(t)$ is the motor torque, $F^\mathrm{B}_i(t)$ the friction brake force, $\omega_i(t) = c^\omega v_i(t)$ the motor speed and $x_i(t)=(p_i(t),v_i(t))$, $u_i(t) = (E_i(t), F^\mathrm{B}_{i}(t))$. 
The parameters $c^{E},c^{\omega},c^\mathrm{d}, c^\mathrm{rr}, \omega_i^\mathrm{max}, \mathrm{E}_i^\mathrm{max}$ and $P_i^\mathrm{max}$ are selected as in \cite{Hult2018a_kappa}.
We use $K=100$ and an explicit Runge-Kutta integrator of order $4$ with $\Delta t = 0.2$.
The objective function is 
\begin{multline}
J_{v_i}(y_{v_i})= Q_i^\mathrm{f}(v_{i,K}-v^\mathrm{r})^2 + \\\sum_{k=0}^{K}Q_i(v_{i,k}-v^\mathrm{r})^2+(u_{i,k}-u_i^\mathrm{r})^\top R_i(u_{i,k}-u_i^\mathrm{r}),
\end{multline}
where $v^\mathrm{r}$ is the reference speed, and $u_i^\mathrm{r}$ is an input which maintains the reference speed $v^\mathrm{r}$.
The cost weights are $Q_i = 1/(v_i^\mathrm{r})^2$, $R_i = \mathrm{diag}((1/E_{i}^{\mathrm{max}})^2, 1/F^\mathrm{B,max}_{i})^2)$, with $Q_i^\mathrm{f}$ the cost-to-go associated with the Linear-Quadratic Regulator (LQR) computed with $Q_i,R_i$ and the linearization of \eqref{eq:pdip_dynamicEquations2} around $v_i^\mathrm{r}$.

The vehicles' initial states are selected randomly between $80 \ \mathrm{m}$ and $120 \ \mathrm{m}$ before the intersection, with $v_{i,0}=v_i^\mathrm{r}= 70 \ \mathrm{km/h}$.
The initial solution candidate $w_i^{[0]},T_i^{[0]}$ has all vehicles driving at $v^\mathrm{r}$ at all times $k = 0, \hdots,K$, and Lagrange multipliers and slacks $\lambda^{[0]} =0$, $\mu^{[0]}=s^{[0]}=\mathbf{1}$, with $\tau^{[0]}=1$.
\iffigs
\begin{figure*}[t]
	\begin{subfigure}[t]{0.32\linewidth}
		\centering
		\includegraphics[width=0.98\linewidth,trim= 0 0 0 0]{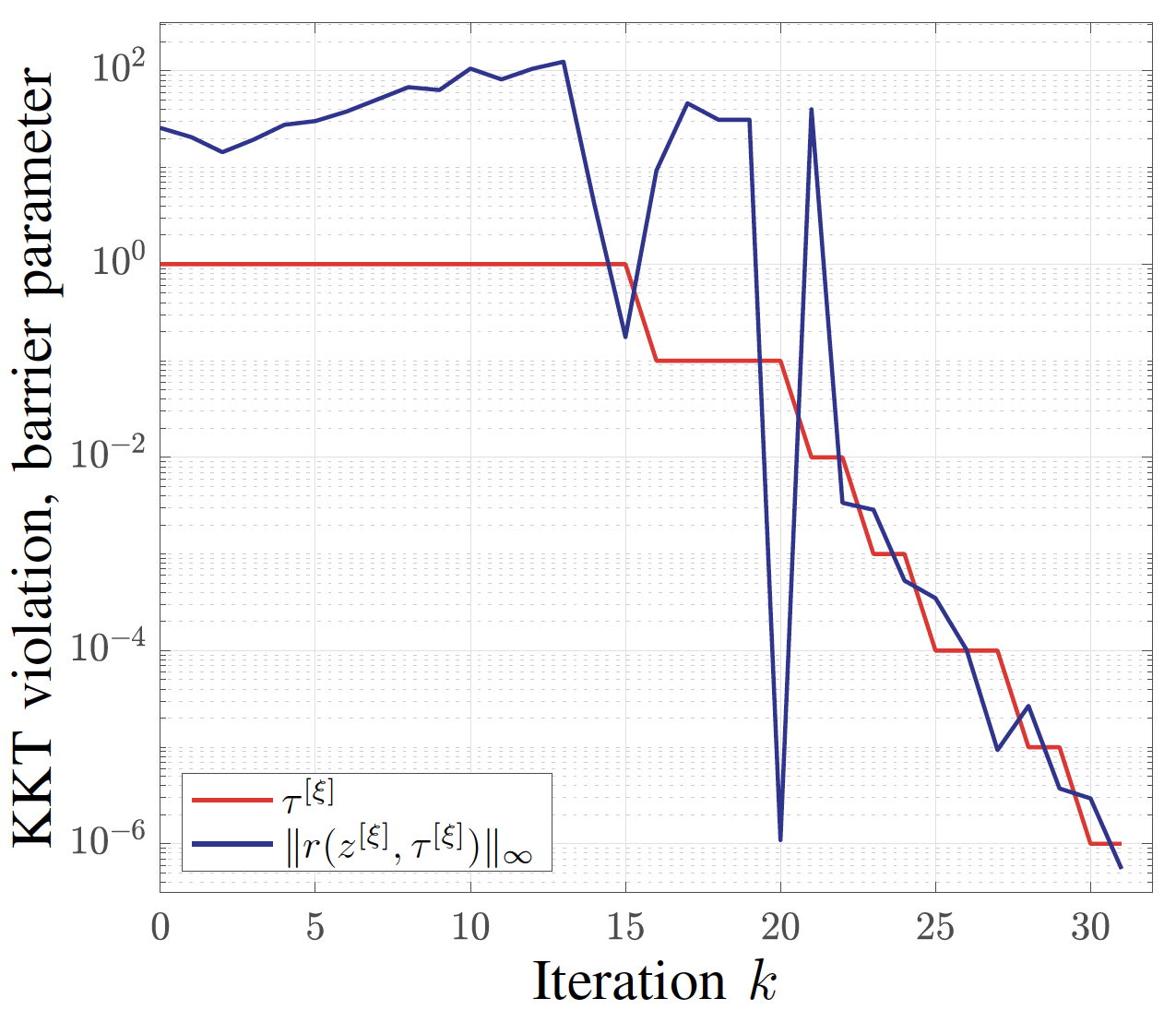}
		\caption{Progress Metrics, tolerance in dashed line.}\label{fig:pdip_exampleProgress}
	\end{subfigure}
	\hspace{0.5em}
	\begin{subfigure}[t]{0.32\linewidth}
		\centering
		\includegraphics[width=0.96\linewidth,trim= 0 20 0 0]{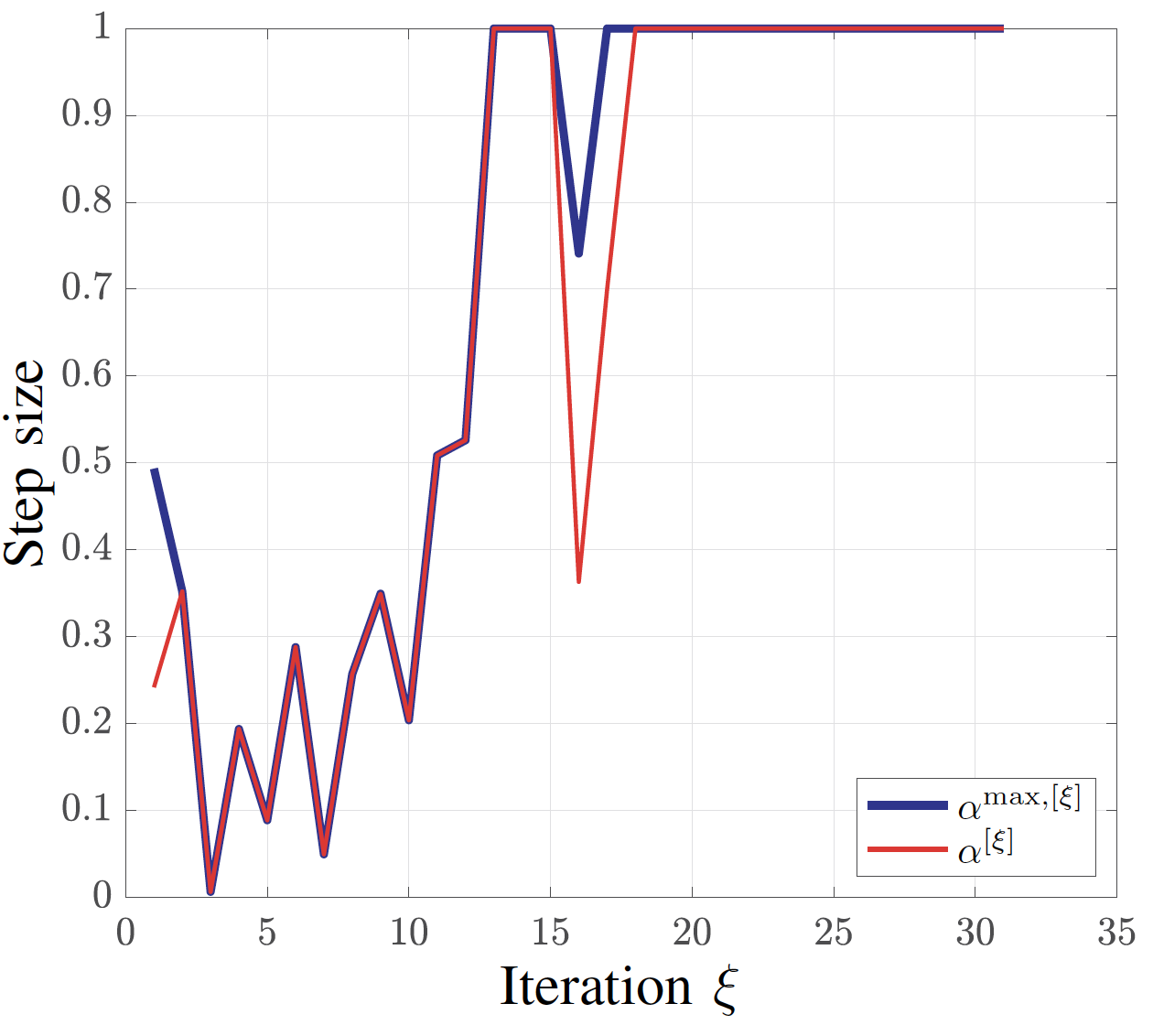}
		\caption{Maximum step size satisfying the positivity constraints~\eqref{eq:pdip_fractionOfTheBoundary} and selected step size.}\label{fig:pdip_exampleStepSize}
	\end{subfigure}
	\hspace{0.5em}
	\begin{subfigure}[t]{0.32\linewidth}
		\centering
		\includegraphics[width=0.95\linewidth]{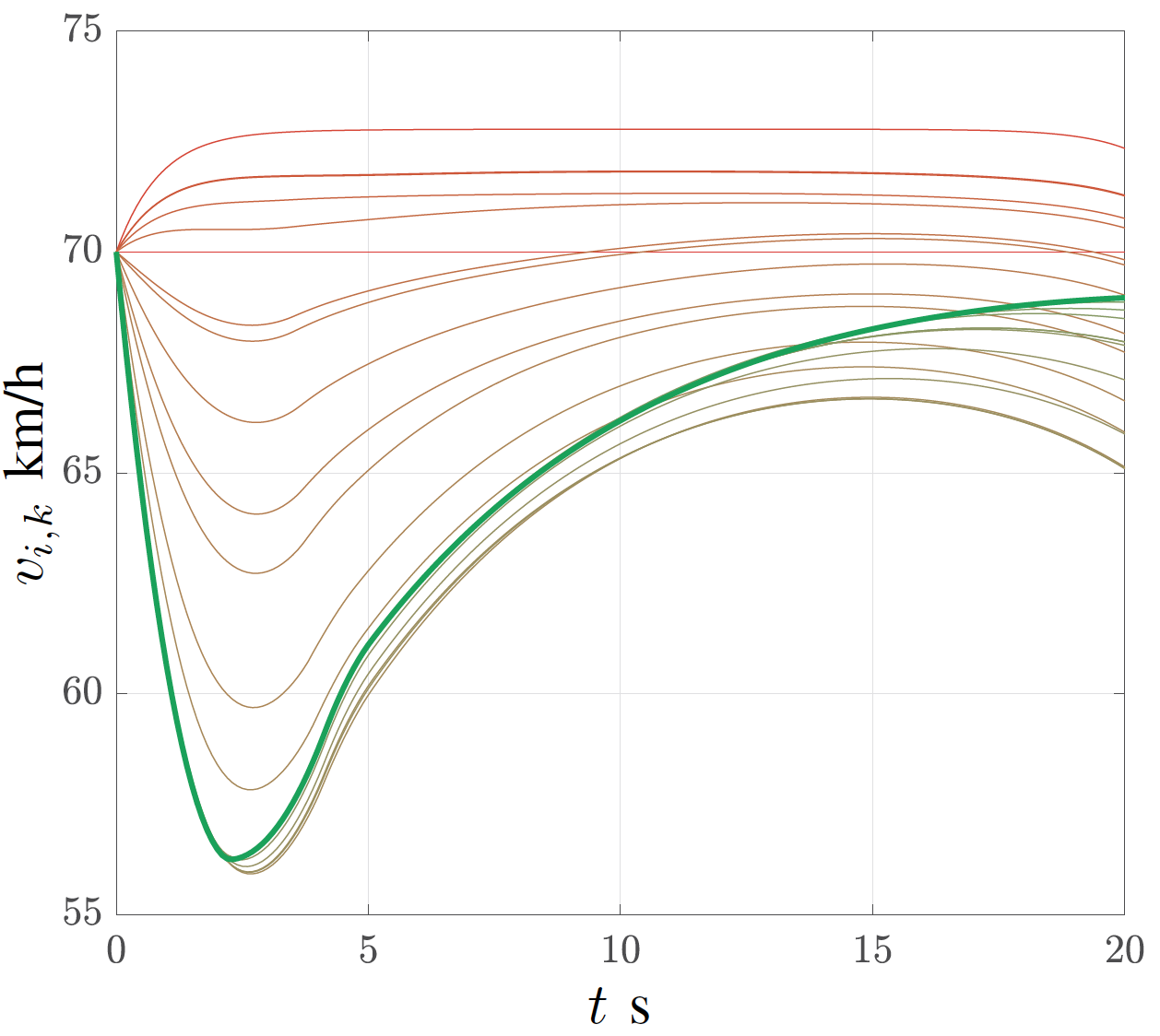}
		\caption{The velocity component of the solution for one of the 12 vehicles over the iterates $\kiter$.}\label{fig:pdip_exampleVelocity}
	\end{subfigure}
	\caption{Data from a 12-Vehicle Example. In (a) the increases in $\|r(z)\|$ follow decreases in $\tau$. In (c) the initial guess is displayed in red, the optimal solution in thick green, and the intermediate iterates in intermediate hues.}\label{fig:pdip_example}
\end{figure*}
\fi

The evolution over the iterates of  $\left \|r(z)^{[\kiter]}\right \|_\infty$ 
 and $\tau^{[\kiter]}$ is shown in \fig{\ref{fig:pdip_exampleProgress}}, and the step size is shown in \fig{\ref{fig:pdip_exampleStepSize}}.

As an illustration of the algorithm's progression in the primal variables, \fig{\ref{fig:pdip_exampleVelocity}} shows the velocity profile of one of the 12 vehicles, at each iterate. We observe that the final 15 iterates all are similar enough to the solution to be considered identical in a practical context.

For illustration,  the sparsity-pattern of  $M$ is given in \fig{\ref{fig:pdip_exampleKKTMatrix}}. The size of $M$ is $25832\times 25832$, where  $M^\mathrm{v}$ is $24176 \times 24176$,  $M^\mathrm{L}_{l}$ is $404 \times 404$ and  $M^\mathrm{C}$ is $40 \times 40$. 
Besides evaluating the involved functions and their derivatives, the main computational effort is therefore  the factorization of the vehicle blocks $M^\mathrm{v}_i$, each roughly sized $2010\times2010$, where the variation in size depends on the fact that the vector of times $T_i$ is different for each vehicle $i$.
Note that the factors for all $M^\mathrm{v}_i$ can be computed in parallel between the vehicles (Line~\ref{alg:ap_distributionLevel2:vehicleComp1} in Algorithm~\ref{alg:ap_distributionLevel2}), and $\bar M^\mathrm{L}_l$ can be factorized in parallel between the lanes (Line~\ref{alg:ap_distributionLevel2:laneComp1}),
therefore greatly reducing the computation time through parallelization.

However, the computational time does not necessarily dominate the time it takes to perform one iterate.
In~\cite{Hult2018c_kappa} the time required to communicate between the vehicles, lane-centers and intersection-center was observed to be orders of magnitude larger. While that was partially implementation-dependent, there are some intrinsic limitations (e.g., packet losses, communication being serial or only partially parallel) which suggest that communication would be the bottleneck in a practical context.
In the next section, we analyze the communication requirements, and discuss some modifications to the scheme which decrease both the number of transmissions and the amount of data communicated.
\iffigs
\begin{figure}[t]
	\centering
	\includegraphics[width=0.67\linewidth]{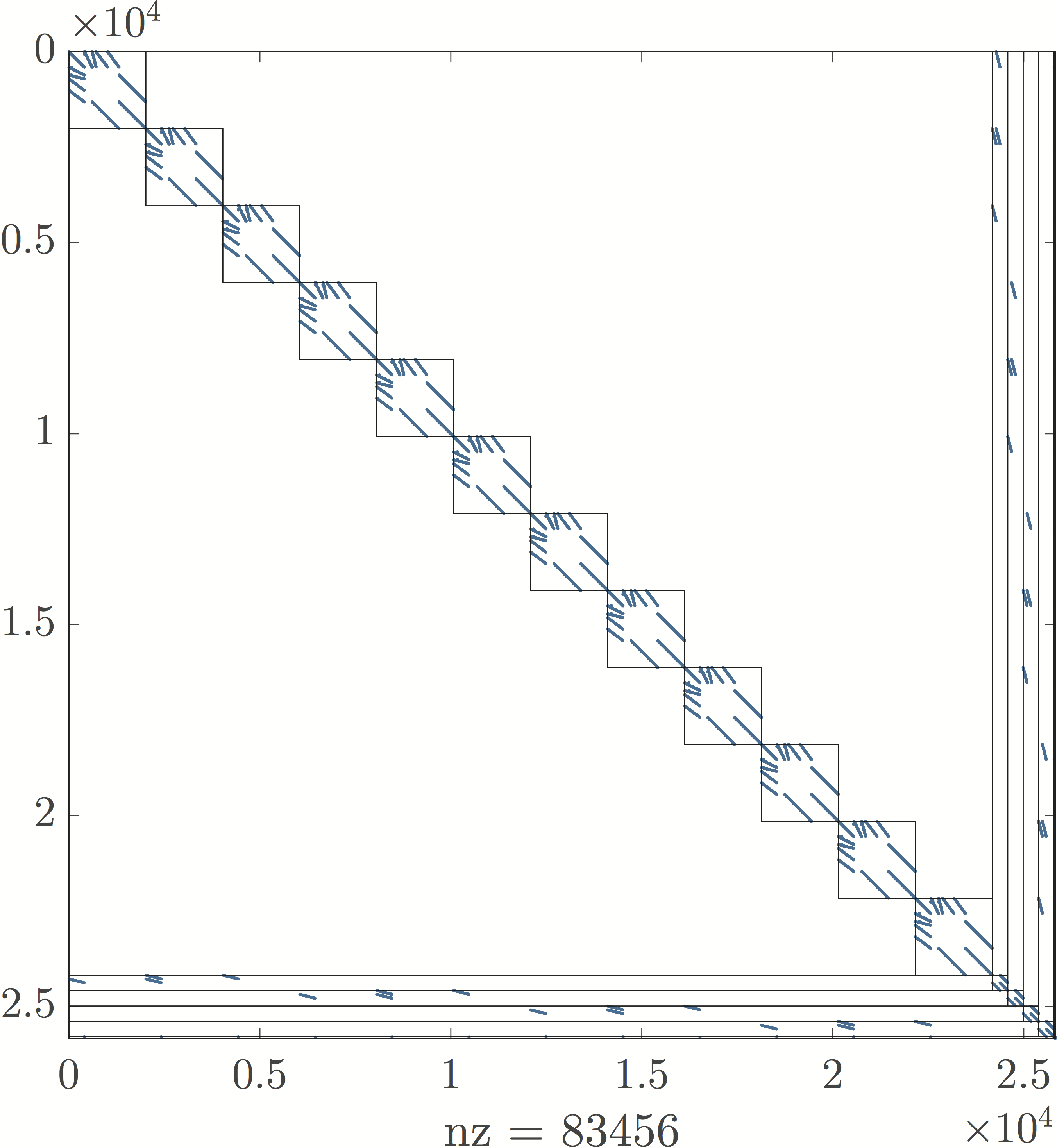}
	\caption[Sparsity pattern of KKT-matrix]{KKT-matrix $M(z)$ from a 12-vehicle scenario. The large upper left hand block is $M^\mathrm{v}$, consisting of the sub-blocks $M^\mathrm{v}_{i}$, $i\in \mathcal{N}$. The smaller blocks in the lower right corner are $M^\mathrm{L}_{l}$, while $M^\mathrm{C}$ is so small that it is essentially not visible. The lines demarcate the sections of $M^\mathrm{vL}$ and $M^\mathrm{vC}$ associated with the RECA constraints on each lane and, essentially not visible, the SICA constraints }\label{fig:pdip_exampleKKTMatrix}
\end{figure}
\fi

\section{Communication Requirements}\label{sec:ap_communication}
In this section, we discuss the communication requirements of Algorithm~\ref{alg:ap_distributedIP}.
We first analyze the data flow between the vehicles, lane centers and intersection center in Section~\ref{sec:pdip_communicationAnalysis}, and discuss how the data exchange required by the proposed algorithm can be reduced in Sections~\ref{sec:pdip_communicationDataReduction} and \ref{sec:pdip_communicationRoundReduction}. 

\subsection{Analysis of Communication Requirements}\label{sec:pdip_communicationAnalysis}

\begin{figure*}[t]
	\centering
	\begin{overpic}[width=1\linewidth]{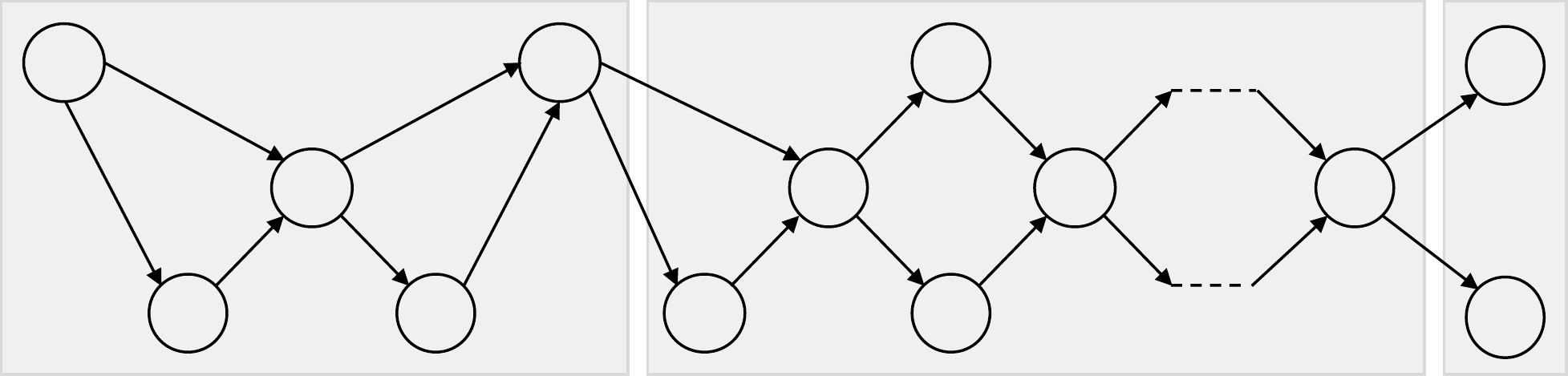}
		\put(3.3,19.5){$\mathrm{v}_i$}
		\put(34.8,19.5){$\mathrm{v}_i$}
		\put(59.8,19.5){$\mathrm{v}_i$}
		\put(95.2,19.5){$\mathrm{v}_i$}
		\put(11.3,3.3){$\mathrm{L}_l$}
		\put(27,3.3){$\mathrm{L}_l$}
		\put(44.2,3.3){$\mathrm{L}_l$}
		\put(59.8,3.3){$\mathrm{L}_l$}
		\put(95.2,3.3){$\mathrm{L}_l$}
		\put(19,11.2){$\mathrm{C}$}
		\put(52,11.2){$\mathrm{C}$}
		\put(67.8,11.2){$\mathrm{C}$}
		\put(85.6,11.2){$\mathrm{C}$}
		\put(11,18){{$\mathcal{D}_{\mathrm{v}_i, \mathrm{C}}$}}
		\put(0.5,12){{$\mathcal{D}_{\mathrm{v}_i, \mathrm{L}_{l(i)}}$}}
		\put(11,8){{$\mathcal{D}_{\mathrm{L}_l,\mathrm{C}}$}}
		\put(23,17){$\Delta \mu^\mathrm{Cv}_i$}
		\put(19,6.5){$\Delta \mu^\mathrm{CL}_l$}
		\put(27.5,11){$\Delta \mu^\mathrm{Lv}_{l,i}$}
		\put(37,11){$\Delta p_i$}
		\put(41.5,21.6){$\Delta T_i, \alpha^\mathrm{max}_{\mathrm{v}_i}$}
		\put(41.5,19){$\phi^\mathrm{v}_i(0)$, $\mathrm{D}\phi^\mathrm{v}_i(0)$}
		\put(48.5,6.8){$\alpha^\mathrm{max}_{\mathrm{L}_l}$}
		\put(48.5,3.8){$\phi^\mathrm{L}_l(0)$}
		\put(48.5,0.8){$\mathrm{D}\phi^\mathrm{L}_l(0)$}
		\put(55.5,9.5){$\alpha$}
		\put(55.5,13.5){$\alpha$}
		\put(64,6.5){ $\phi^\mathrm{L}_l(\alpha)$}
		\put(64,17){$\phi^\mathrm{v}_{i}(\alpha)$}
		\put(64,19.6){$\left \|r^\mathrm{v}_i\right \|_\infty$}
		\put(64,3.5){$\left \|r^\mathrm{L}_l\right \|_\infty$}
		\put(71.5,9.5){$\alpha$}
		\put(71.5,13.5){$\alpha$}
		\put(82,6.5){$\phi^\mathrm{L}_l(\alpha)$}
		\put(82,17){$\phi^\mathrm{v}_{i}(\alpha)$}
		\put(82,19.6){$\left \|r^\mathrm{v}_i\right \|_\infty$}
		\put(82,3.5){$\left \|r^\mathrm{L}_l\right \|_\infty$}
		\put(89.8,9.5){$\tau$}
		\put(89.8,13.5){$\tau$}
		\put(10,25){\small \textsc{SearchDirection}}
		\put(53,25){\small \textsc{StepSizeSelection}}
		\put(91,25){\small New It.}
	\end{overpic}
	\caption{Illustration of the communication flow in the problem. The horizontal direction indicates the order in which the communication is done whereas the vertical differentiates the vehicle, lane and intersection levels. 
		With $\|r^\mathrm{v}_i\|_\infty$, $\|r^\mathrm{L}_l\|_\infty$ we denote the residual norms obtained with step size $\alpha$.
	}\label{fig:pdip_comFlow}
\end{figure*}

\begin{figure*}
\begin{center}
	\footnotesize
			\begin{tabular}{lclcc}
				\hline
				Link & \multicolumn{2}{l}{\hspace{0.3em}Location} &Data per Communication Round & \# Floats		\\
				\hline
				\vspace{-0.8em}
				\\
				& &&\textsc{SearchDirection} \\
				\hline
				\vspace{-0.5em}
				\\
				$\mathrm{v}_i$ to $\mathrm{L}_l$&A.\ref{alg:ap_distributionLevel2} & L.\ref{alg:ap_distributionLevel2:vehicleComp1} & 
				\hspace{1em}
				$\underbrace{M^\mathrm{Lv}_{l,i} \left (M^\mathrm{v}_i\right )^{-1} M^\mathrm{vL}_{i,l}}_{(K+1)\times (K+1),  \ \text{symmetric}}$, 
				$\underbrace{M^\mathrm{Cv} \left (M^\mathrm{v}_{i}\right )^{-1} M^\mathrm{vL}_{i,l} }_{n_{T_{i}}\times (K+1)}$, 
				$\underbrace{M^\mathrm{Lv}_{l,i} \left (M^\mathrm{v}_{i}\right )^{-1}r^\mathrm{v}_{i}}_{K+1}$, 
				$\underbrace{p^\mathrm{v}_{i}}_{K+1}$ 
				\hspace{1em}
				&$\frac{1}{2}K^2+\left(n_{T_{i}}+\frac{7}{2}\right)K+n_{T_{i}}+3$\\[2.5em]
				$\mathrm{v}_i$ to $\mathrm{C}$&A.\ref{alg:ap_distributionLevel2}& L.\ref{alg:ap_distributionLevel2:vehicleComp1} &
				$\underbrace{M^\mathrm{Cv} \left (M^\mathrm{v}_{i}\right )^{-1} M^\mathrm{vC}}_{n_{T_{i}}\times n_{T_{i}}, \ \text{symmetric}}$, 
				$\underbrace{M^\mathrm{Cv}  \left (M^\mathrm{v}_i\right )^{-1}r^\mathrm{v}_i}_{n_{T_i} }$, 
				$\underbrace{T_i}_{n_{T_i}}$
				& $\frac{1}{2}n_{T_i}^2 + \frac{5}{2}n_{T_i}$\\[2.5em]
				$\mathrm{L}_l$ to $\mathrm{C}$&A.\ref{alg:ap_distributionLevel2}& L.\ref{alg:ap_distributionLevel2:laneComp1}  &
				$\underbrace{\bar{M}^\mathrm{LC}_l \left ( \bar{M}^\mathrm{L}\right )^{-1} \bar{M}^\mathrm{CL}_l}_{n_{T_l^\mathrm{L}}\times n_{T_l^\mathrm{L}}, \ \text{symmetric}}$, $\underbrace{\bar{M}^\mathrm{LC}_l \left ( \bar{M}^\mathrm{L}\right )^{-1}\bar{r}^\mathrm{L}_l}_{n_{T_l^\mathrm{L}}}$ &$\frac{1}{2}n_{T_l^\mathrm{L}}^2+\frac{3}{2}n_{T_l^\mathrm{L}} $\\[0.5em]
				$\mathrm{C}$ to $\mathrm{L}_l$&A.\ref{alg:ap_distributionLevel2}& L.\ref{alg:ap_distributionLevel2:centralComp} & $\Delta \mu^\mathrm{CL}_l$& $n_{T_l^\mathrm{L}}$\\[0.5em]
				$\mathrm{C}$ to $\mathrm{v}_i$&A.\ref{alg:ap_distributionLevel2}& L.\ref{alg:ap_distributionLevel2:centralComp} &$\Delta \mu^\mathrm{Cv}_{i}$& $n_{T_i}$\\[0.5em]
				$\mathrm{L}_l$ to $\mathrm{v}_i$&A.\ref{alg:ap_distributionLevel2}& L.\ref{alg:ap_distributionLevel2:laneComp2}&$\Delta \mu^\mathrm{Lv}_{l,i}$& $2K$ \\[1em]
				& &&\textsc{StepSizeSelection} \\
				\hline
				\vspace{-0.5em}
				\\
				$\mathrm{v}_i$ to $\mathrm{L}_l$ &A.\ref{alg:ap_stepSize}& L.\ref{list:pdip_alphamaxVehicle} &$\Delta p_{i}$& $K+1$\\[0.5em]
				$\mathrm{v}_i$ to $\mathrm{C}$ &A.\ref{alg:ap_stepSize}& L.\ref{list:pdip_alphamaxVehicle}  &$\alpha_{\mathrm{v}_i}^{\mathrm{max}}$, $\Phi^\mathrm{v}_{i}(0)$, $\mathrm{D}\Phi^\mathrm{v}_{i}(0)$, $\Delta T_{i}$ & $3+n_{T_{i}}$ \\[0.5em]
				$\mathrm{L}_l$ to $\mathrm{C}$ &A.\ref{alg:ap_stepSize}& L.\ref{list:pdip_alphamaxLane} &$\alpha_{\mathrm{L}_l}^{\mathrm{max}}$, $\Phi^\mathrm{L}_l(0)$, $\mathrm{D}\Phi^\mathrm{L}_{l}(0)$& $3$ \\[0.5em]
				$\mathrm{C}$ to $\mathrm{v}_i,\,\mathrm{L}_l$ &A.\ref{alg:ap_stepSize}& L.\ref{alg:pdip_stepSize:resetAlpha}&$\alpha $& $1$ \\[0.5em]
				$\mathrm{v}_i$ to $\mathrm{C}$ &A.\ref{alg:ap_stepSize}& L.\ref{alg:ap_StepSize:vehicleMeritFunction} &$\Phi^\mathrm{v}_{i}(\alpha )$& $1$\\[0.5em]
				$\mathrm{L}_l$ to $\mathrm{C}$ &A.\ref{alg:ap_stepSize} & L.\ref{alg:ap_StepSize:laneMeritFunction} &$\Phi^\mathrm{L}_{l}(\alpha )$& $1$\\[0.5em]
				\hline
		\end{tabular}
\captionof{table}{Summary of the communication between the vehicles $\mathrm{v}_i$,  lane centers $\mathrm{L}_l$ and  intersection center $\mathrm{C}$. The second column refers to the Algorithm (A) and Line (L) where the communication occurs. The numbers under braces denote the amount of non-zeros of the object. 
}\label{tab:pdip_communicationSummary}
\end{center}
\end{figure*}
Most data is exchanged during the solution of the KKT-system in Algorithm~\ref{alg:ap_distributionLevel2} and the selection of the step-size in Algorithm~\ref{alg:ap_stepSize}.
Descriptions of the data involved as well as the number of floats communicated are summarized in Table~\ref{tab:pdip_communicationSummary}.

Most often $K \gg n_{T_i}$, whereby most communication occurs during Line~\ref{alg:ap_distributionLevel2:vehicleComp1} of Algorithm~\ref{alg:ap_distributionLevel2} when the vehicle sends~$\mathcal{D}_{\mathrm{v}_i, \mathrm{L}_{l(i)}}$.
Besides the communication between the lane-centers and the intersection-center and an initial round of communication where the initial guess is sent to the vehicles (lines~\ref{alg:ap_distributedIP:centerIG}-\ref{alg:ap_distributedIP:laneIG}), the communication required for the remaining parts (i.e., the indication of a new iteration, the current barrier parameter value, termination of line-search or algorithm completion)  consists of  single floats and logicals.
As illustrated in \fig{\ref{fig:pdip_comFlow}}, these can be sent together with the search direction and step size results. 

\subsubsection*{Communication in the Example}
In the example  $K=100$, whereby each vehicle sends more than 5000 floats \emph{per iterate} (more than 320000 bits) to their respective lane-center.
Even if all vehicles communicate in parallel, the physical transmission will take at least $58.7$ ms using the 802.11p protocol, i.e., the current standard for vehicular communications. The time per bit is computed assuming double precision and using~\cite{Fernandez2012_kappa} 
\begin{align}
	\label{eq:comm_time}
	50+8\, \mathrm{ceil}((n_{\text{data bits}}+22)/48) \ \mu \mathrm{s}.
\end{align}
During $33$ iterations, at least $1.94$ s would be spent communicating to construct $\Delta z$, which would be too high in a practical setting.
Next, we discuss how the data exchange can be reduced.

\subsection{Reduction of Data exchange per iterate}\label{sec:pdip_communicationDataReduction}
Most of the communication is due to the dependence on $K^2$ in the number of communicated floats on Line~\ref{alg:ap_distributionLevel2:vehicleComp1} of Algorithm~\ref{alg:ap_distributionLevel2} and corresponds to the enforcement of RECA. 
Unfortunately, reducing the horizon length $K$ to contain this issue would result in a significant performance loss.
As an alternative to tackle this problem, we propose to replace all RECA constraints \eqref{eq:discreteRECA}
with the approximation
\begin{subequations}\label{eq:pdip_parametrizedRECA}
	\begin{align}
	p_{i,k} + \delta_{i}/2  &\leq \rho_{i,k}(\theta_{i}),&&k =1, \hdots, K\label{eq:pdip_parametrizedRECA1}\\
	\rho_{i,k}(\theta_{i}) +\delta_{i}/2&\leq  p_{i+1,k},&&k =1, \hdots, K,\label{eq:pdip_parametrizedRECA2}
	\end{align}
\end{subequations}
where $\rho_{i,k}(\theta_{i})$ is a function of $k$, and \emph{coupling parameters} $\theta_{i} \in \mathbb{R}^{q}$, introduced as additional decision variables in the fixed order problem \eqref{eq:fixedOrderProblem}.
Equation~\eqref{eq:pdip_parametrizedRECA} enforces RECA constraints indirectly, by requiring that function $\rho_{i,k}(\theta_{i})$ be between $p_{i,k}$ and $p_{i+1,k}$ at all $k$. 
In this case, the amount of data to be exchanged scales with $q^2$ rather than on $K^2$, therefore allowing one to significantly reduce the amount of data to be exchanged.
The parameterized RECA coupling can be included in the distributed scheme in two different ways, as we detail next.

\begin{figure}[t]
	\begin{center}
		\begin{overpic}[width=0.9\linewidth]{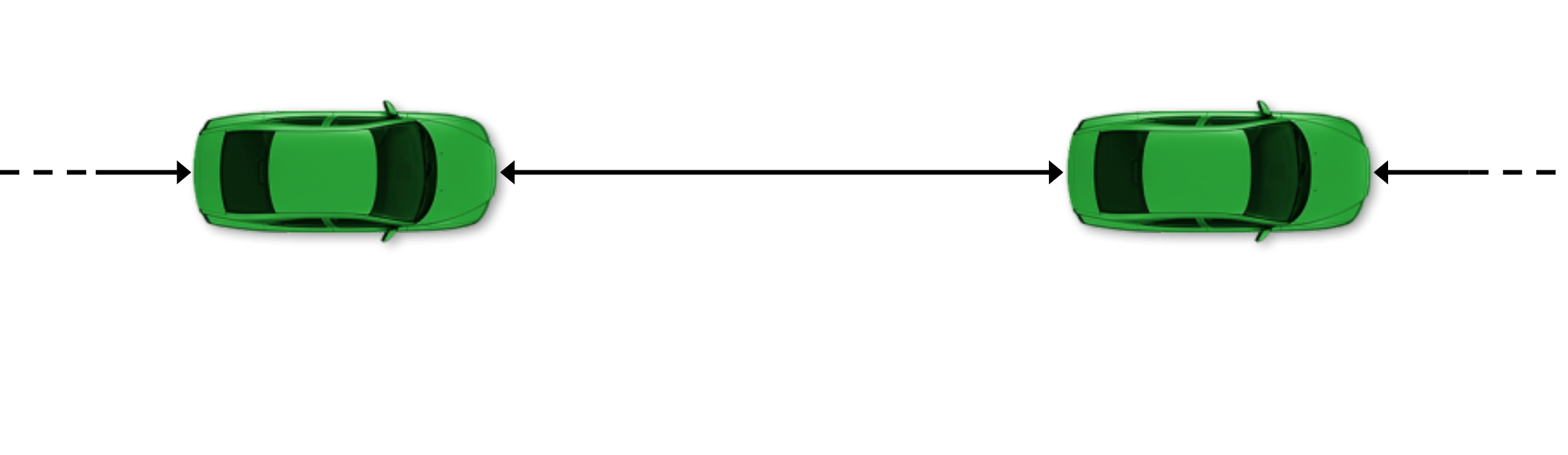}
			\put(39.5,12){\footnotesize $\theta_{i,i} = \theta_{i} = \theta_{i+1,i}$}
			\put(34,7){\footnotesize $\rho_{i}(\theta_{i,i}) = \rho_{i}(\theta_{i+1,i})$}
			\put(43,2){\footnotesize $\mu^\mathrm{Rf}_{i} = \mu^\mathrm{Rr}_{i+1}$}
			\put(20,10){\footnotesize $\mathrm{v}_i$}
			\put(75,10){\footnotesize $\mathrm{v}_{i+1}$}
			\put(31,22){\footnotesize $(\mu^\mathrm{Rf}_{i},s^\mathrm{Rf}_{i}) $}
			\put(-5,22){\footnotesize $(\mu^\mathrm{Rr}_{i},s^\mathrm{Rr}_{i}) $}
			\put(50,22){\footnotesize $(\mu^\mathrm{Rr}_{i+1},s^\mathrm{Rr}_{i+1}) $}
			\put(87,22){\footnotesize $(\mu^\mathrm{Rf}_{i+1},s^\mathrm{Rf}_{i+1}) $}
		\end{overpic}
	\end{center}
	\caption{Illustration of the relationship between the variables introduced in the primal and dual communication reduction approaches. }\label{fig:pdip_recaVarIllustration}
\end{figure}

\subsubsection*{The ``Primal'' Approach}
The first alternative is to handle the coupling parameters $\theta_{i}$ at the lane-centers, and constraints \eqref{eq:pdip_parametrizedRECA1}, \eqref{eq:pdip_parametrizedRECA2}  on-board the vehicles, i.e., the problem to be solved becomes
	\begin{align*}
	\underset{y}{\min} & \quad \sum_{i=1}^{N}J_{i}(w_{i}) &&& \\
	\mathrm{s.t.} 
	& \quad  g_{i}(w_i,T_i) = 0, &i &\in \mathcal{I}^\mathrm{v}, &&|\, \lambda_i,  \\
	& \quad  h_{i}^\mathrm{PR}(w_i,\theta_{i-1},\theta_i) \leq 0, & i &\in \mathcal{I}^\mathrm{v}, &&|\, \mu_i^\mathrm{PR}, \\
	& \quad h^\mathrm{C}(T) \leq 0, &&&&|\, \mu^\mathrm{C},  
	\end{align*}
with  
\begin{align*}
h_{i}^\mathrm{PR} (w_i,\theta_{i-1},\theta_{i}) = 
\begin{bmatrix}
c_i(x_{i,k},u_{i,k}) \\ 
p_{i,k} + \delta_i/2 - \rho_{i,k}(\theta_i) \\
\rho_{i-1,k}(\theta_{i-1}) + \delta_{i-1}/2 - p_{i,k}  \\
\end{bmatrix}.
\end{align*}
The vehicle variables $z^\mathrm{v}_{i}$ now include additional multipliers and the corresponding slack variables, as $\mu^\mathrm{PR}_{i,k}=\left (\mu^\mathrm{P}_{i,k},\mu^\mathrm{Rf}_{i,k},\mu^\mathrm{Rr}_{i,k}\right )$.

The lane-center variables become $z^\mathrm{L}_{l}=(\theta_{1},\ldots,\theta_{N^\mathrm{v}_l})$, and
\begin{align*}
	r^\mathrm{L}_{l}=\nabla_{\theta_{i}}\mathcal{L}= 
	\left (\nabla_{\theta_{i}}\rho_{i}\mu^\mathrm{Rr}_{i+1}-\nabla_{\theta_{i}}\rho_{i}\mu^\mathrm{Rf}_{i}, \ i\in \mathcal{I}^\mathrm{v}_l \right ),
\end{align*}
where $\rho_i = (\rho_{i,0},\ldots,\rho_{i,K})$, and $\mu^\mathrm{R,\cdot}_{i}=\big (\mu^\mathrm{R,\cdot}_{i,0},\ldots,\mu^\mathrm{R,\cdot}_{i,K+1} \big)$.
Then, the corresponding KKT matrix components become
\begin{align*}
	M^\mathrm{L}_{l,i} &= M^\mathrm{Lr}_{l,i} + M^\mathrm{Lf}_{l,i}\\
	M^\mathrm{Lr}_{l,i} &= \sum_{i\in\mathcal{I}^\mathrm{v}_l} \left \langle \nabla^2_{\theta_{i}} \rho_{i},\mu^\mathrm{Rr}_{i+1} \right \rangle, \quad M^\mathrm{L,f}_{l,i} = -\sum_{i\in\mathcal{I}^\mathrm{v}_l} \left \langle \nabla^2_{\theta_{i}} \rho_{i},\mu^\mathrm{Rf}_{i} \right \rangle, \\
	M^\mathrm{Lv}_{l,i} &= \nabla_{\theta_{i-1}}\rho_{i-1}\nabla_{\mu^\mathrm{Rr}_{i}} z^\mathrm{v}_{i}  - \nabla_{\theta_{i}}\rho_{i}\nabla_{\mu^\mathrm{Rf}_{i}}  z^\mathrm{v}_{i}.
\end{align*}
In this case, if $M^\mathrm{L} \nsucc 0$, one might need to introduce additional Hessian regularization at the lane center level (c.f. the discussion in Section~\ref{sec:pdip_nonconvexity}). This operation can be done by each lane center independently of the other lane centers and the single vehicles.

The size of the information assembled by each vehicle and sent to the lane center and central node is 
	\begin{align*}
	&M^\mathrm{Lv}_{l,i} \left (M^\mathrm{v}_{i}\right )^{-1} M^\mathrm{vL}_{i,l},&& \text{size: } q \times q \text{ twice},\\
	&M^\mathrm{Lv}_{l,i} \left (M^\mathrm{v}_{i}\right )^{-1} r^\mathrm{v}_{i},&& \text{size: } 2q ,\\
	&M^\mathrm{Lv}_{l,i} \left (M^\mathrm{v}_{i}\right )^{-1} M^\mathrm{vC}_{i},&& \text{size: } 2q \times n_{T_i},\\
	&\nabla_{\theta_{i-1}}\rho_{i-1}\mu^\mathrm{R,r}_{i}-\nabla_{\theta_{i}}\rho_{i}\mu^\mathrm{R,f}_{i},&& \text{size: } 2q .
	\end{align*}
This amounts to $q^2 + (5+2n_{T_i})q + 2$ floats on  Line~\ref{alg:ap_distributionLevel2:vehicleComp1} of Algorithm~\ref{alg:ap_distributionLevel2}. Moreover, the information  sent from the lane center to each vehicle becomes $\Delta \theta_{i}$, $\Delta \theta_{i-1}$ ($2q$ floats). Note that, while $M^\mathrm{Lv}_{l,i}$ is available at the lane center and not in each vehicle, that matrix only has two nonzero entries equal to plus or minus one and remains constant throughout the iterates, such that it can cheaply be sent to each vehicle at the beginning of each NLP solution. The same applies to $M^\mathrm{vC}_{i}$, which is originally in the central node.

\subsubsection*{The ``Dual" Approach}
Alternatively, one can introduce as optimization variables one copy of $\theta_{i-1}, \theta_{i}$ for each vehicle as $\theta_{i,i-1}$ and $\theta_{i,i}$. It is then also required to introduce the additional coupling constraint 
$\theta_{i,i}-\theta_{i+1,i}=0$.
The optimization problem then reads as
	\begin{align*}
	\underset{y}{\min} & \quad \sum_{i=1}^{N}J_{i}(w_{i}) &&& 
	\\
	\mathrm{s.t.} 
	& \quad  g(w_i,T_i) = 0, &i &\in \mathcal{I}^\mathrm{v}, &&|\, \lambda_i, 
	\\
	& \quad h_{i}^\mathrm{PR}(w_i,\theta_{i,i-1},\theta_{i,i}) \leq 0, & i &\in \mathcal{I}^\mathrm{v}, &&|\, \mu_i^\mathrm{PR},
	\\
		& \quad \theta_{i,i} - \theta_{i+1,i} = 0, &
		 i & \in \mathcal{I}^\mathrm{v}_l, \, l\in\mathcal{I}^\mathrm{L} ,  &&|\, \lambda_{i}^\theta,
		\\
	& \quad h^\mathrm{C}(T) \leq 0, &&&&|\, \mu^\mathrm{C},  
	\end{align*}
where $\mu^\mathrm{PR}_{i,k}=\left (\mu^\mathrm{P}_{i,k},\mu^\mathrm{Rf}_{i,k},\mu^\mathrm{Rr}_{i,k}\right )$ and
\begin{align*}
h_{i}^\mathrm{PR} (w_i,\theta_{i,i-1},\theta_{i,i}) := 
\begin{bmatrix}
c_i(x_{i,k},u_{i,k}) \\ 
p_{i,k} + \delta_i/2 - \rho_{i,k}(\theta_{i,i}) \\
\rho_{i-1,k}(\theta_{i,i-1}) + \delta_{i-1}/2 - p_{i,k}  \\
\end{bmatrix}.
\end{align*}

In this approach, $(\theta_{i,i-1}, \theta_{i,i}, \mu^\mathrm{Rf}_{i}, \mu^\mathrm{Rr}_{i}, s^\mathrm{Rf}_{i}, s^\mathrm{Rr}_{i})$ are included in the vehicle variables $z^\mathrm{v}_i$. In the lane centers, we have $z^\mathrm{L}_{l}=(\lambda_i^\theta, i\in\mathcal{I}_l^\mathrm{v})$
and $r^\mathrm{L}_{l}=(\theta_{i,i} - \theta_{i+1,i}, \, i  \in \mathcal{I}^\mathrm{v}_l, \, l\in\mathcal{I}^\mathrm{L})$.
In the Dual approach, all primal variables are at the vehicle level, and block-wise Hessian regularization needs to be done at the vehicle level only (c.f. the discussion in Section~\ref{sec:pdip_nonconvexity}).

The information sent from vehicle $i$  on Line~\ref{alg:ap_distributionLevel2:laneComp2} of Algorithm~\ref{alg:ap_distributionLevel2} is then
	\begin{align*}
	&\frac{\partial z^\mathrm{v}_i}{\partial \theta_{i,\cdot}}^\top \left (M^\mathrm{v}_i\right )^{-1}\frac{\partial z^\mathrm{v}_i}{\partial \theta_{i,\cdot}},&& \text{size } q\times q,\\
	&\frac{\partial z^\mathrm{v}_i}{\partial \theta_{i,\cdot}}^\top \left (M^\mathrm{v}_i\right )^{-1}r^\mathrm{v}_i,&& \text{size } q,\\
	&\frac{\partial z^\mathrm{v}_i}{\partial \theta_{i,\cdot}}^\top \left ( M^\mathrm{v}_i\right )^{-1}\frac{\partial z^\mathrm{v}_i}{\partial T_{i}},&& \text{size } q\times n_{T_i},\\
	&\theta_{i,\cdot},&& \text{size } q,
	\end{align*}
i.e., the same data-amount as the Primal approach, since the information needs to be sent for both $\theta_{i,i-1}, \, \theta_{i,i}$. Note that in this case the same considerations made for $M^\mathrm{Lv}_{l,i}$ and $M^\mathrm{vC}_{i}$ in the Primal approach apply to $\frac{\partial z^\mathrm{v}_i}{\partial \theta_{i,\cdot}}$ and $\frac{\partial z^\mathrm{v}_i}{\partial T_{i}}$. 
Moreover, we have
\begin{align*}
\hspace{-0.3em}\nabla_{w_{i}} \mathcal{L} &= \nabla_{w_{i}} J_{i} + \nabla_{w_{i}}g_{i}\lambda_{i} + \nabla_{w_{i}}h^\mathrm{PR}_{i}\mu^\mathrm{PR}_{i} \\
&\hspace{8em}+ \nabla_{\theta_{i,i}} w_i \lambda^\theta_i  -\nabla_{\theta_{i,i-1}} w_i \lambda^\theta_{i-1},  \\
\hspace{-0.3em}\nabla_{T_{i}} \mathcal{L} &=  \nabla_{T_{i}}g_{i}\lambda_{i} + \nabla_{T_i}h^\mathrm{C}\mu^\mathrm{C}.
\end{align*}
Consequently, the lane-center-to-vehicle communication on Line~\ref{alg:ap_distributionLevel2:laneComp2} of Algorithm~\ref{alg:ap_distributionLevel2} consists of  $\Delta \lambda_{i-1}$, $\Delta \lambda_{i}$, i.e., $2q$ floats.

Since both approaches have the same communication footprint, they are equally valid alternatives.

\subsection{Example}
By selecting $\theta_{i}$ such that $q$ is small, significant reductions in the amount of data communicated are achieved at the cost of some sub-optimality. 
\iffigs
\begin{figure}[t]
	\centering
	\includegraphics[width=0.8\linewidth]{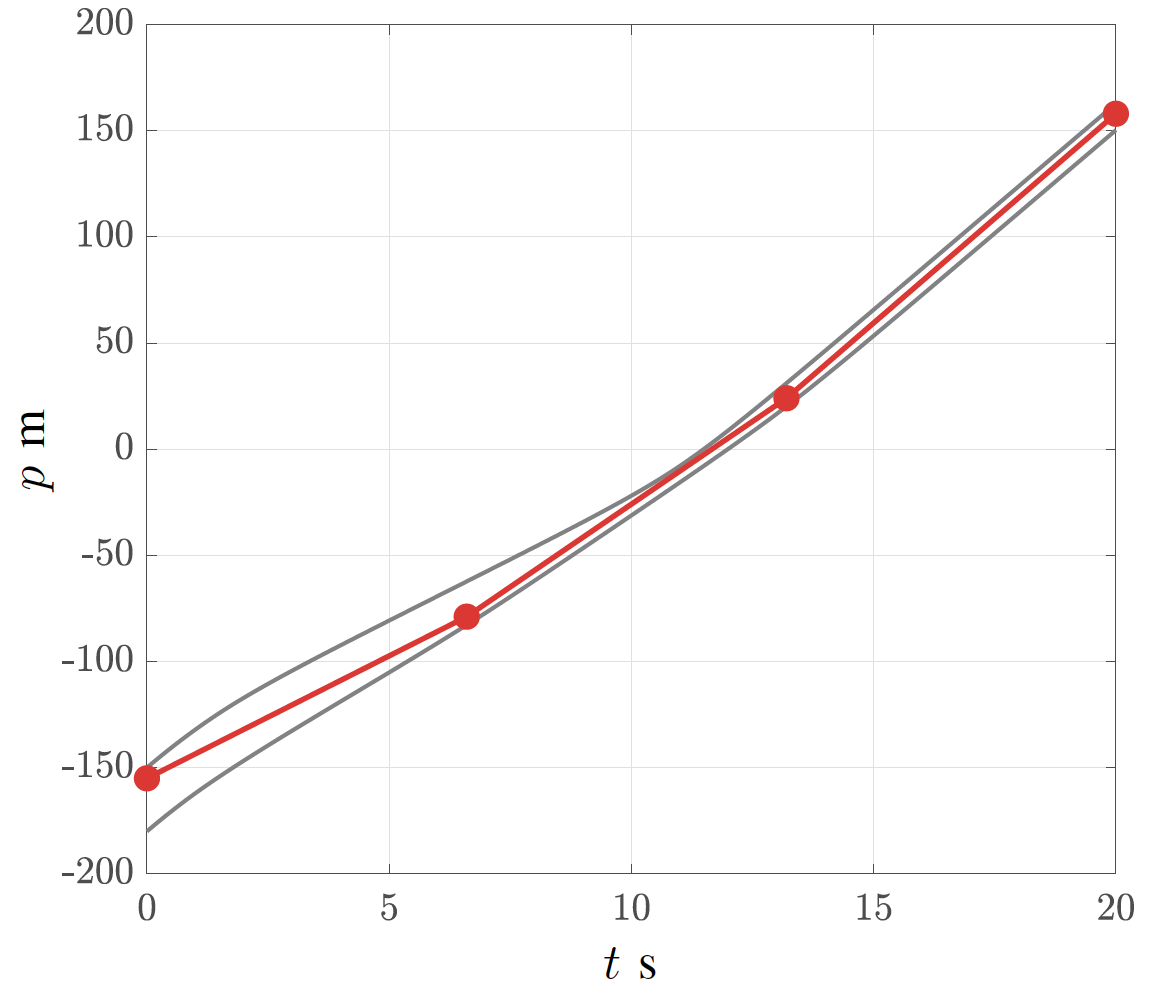}
	\caption{Illustration of a piecewise linear RECA parameterization: vehicle trajectories in gray, $\rho_{i,k}(\theta_{i})$ in red, $\theta_{i}$ as red dots.}\label{fig:pdip_recaApproximationIllustration}
\end{figure}
\fi
To evaluate the trade-off between reduction of communication and cost increase, we consider the case shown in \fig{\ref{fig:pdip_recaApproximationIllustration}}, where  $\rho_{i,k}(\theta_{i})$ is the piecewise linear function
\begin{align*}
&\rho_{i,k}(\theta_{i}) =\\ 
&\hspace{8pt}\left\lbrace \hspace{-4pt}
\renewcommand{\arraystretch}{2}
\begin{array}{ll}
\theta_{i}^{(1)} \hspace{-2pt}+\hspace{-1pt} \frac{\theta_{i}^{(2)}-\theta_{i}^{(1)}}{\rule{0pt}{0.75em}\lfloor K/3\rfloor}k,	& k \in [0, \lfloor K/3\rfloor],\\
\theta_{i}^{(2)} \hspace{-2pt}+\hspace{-1pt} \frac{\theta_{i}^{(3)}-\theta_{i}^{(2)}}{\rule{0pt}{0.75em}\lfloor K/3\rfloor}(k\hspace{-1pt}-\hspace{-1pt}\lfloor K/3\rfloor),	& k \in [\lfloor K/3\rfloor\hspace{-1pt}+\hspace{-1pt}1, 2\lfloor K/3\rfloor],\\
\theta_{i}^{(3)} \hspace{-2pt}+\hspace{-1pt} \frac{\theta_{i}^{(4)}-\theta_{i}^{(1)}}{\rule{0pt}{0.75em}\lceil K/3\rceil}(k\hspace{-1pt}-\hspace{-1pt}2\lfloor K/3\rfloor),	& k \in [2\lfloor K/3\rfloor\hspace{-1pt}+\hspace{-1pt}1, K\hspace{-1pt}+\hspace{-1pt}1], 
\end{array}
\right.
\end{align*}
where the superscript on $\theta_{i}$ indicates the vector element and $q=4$. When $n_{T_{i}}=4$, no more than $60$ floats are sent from a vehicle to the lane center on Line~\ref{alg:ap_distributionLevel2:laneComp2} of Algorithm~\ref{alg:ap_distributionLevel2}, which, according to~\eqref{eq:comm_time}, will take at least $0.7 \ \mathrm{ms}$, i.e., a reduction of almost $99 \%$ in communication time.
\iffigs
\begin{figure}[t]
	\centering
	\includegraphics[width=1\linewidth]{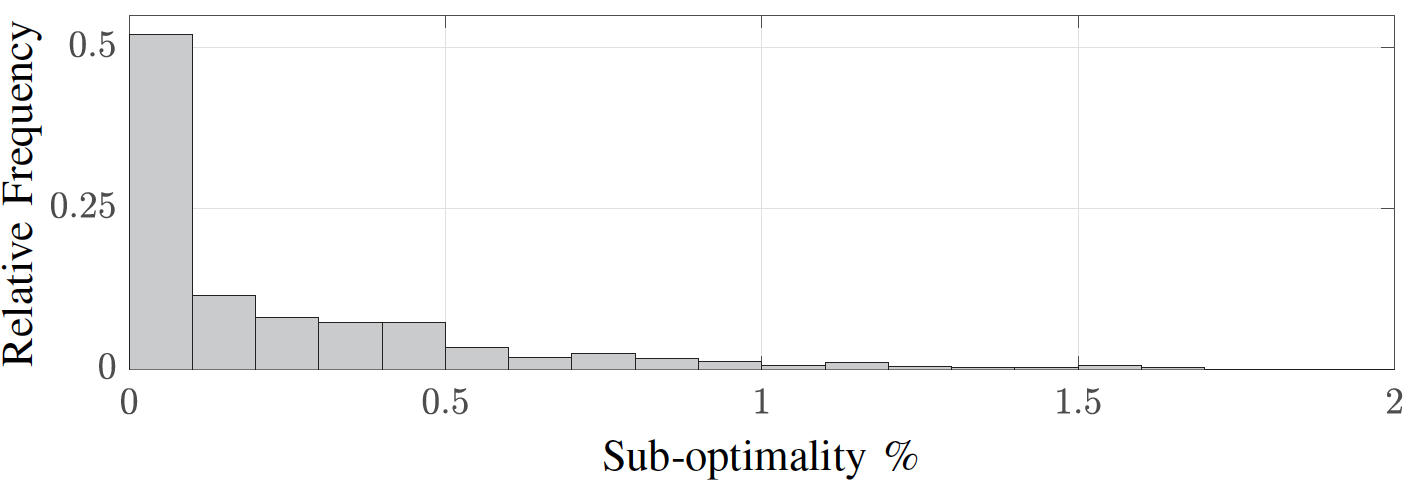}
	\caption{Distribution of the suboptimality resulting from the use of approximate RECA constraints with a piecewise linear $\rho_{i,k}(\theta_{i})$}\label{fig:pdip_approximateRECArandomEval}
\end{figure}
\fi

To assess the sub-optimality induced, we evaluated 500  scenarios with 4 vehicles per lane (16 in total), using the models and objective functions of Section~\ref{sec:ap_example}. 
The vehicles were initialized at randomly drawn distances in the interval $[50, 150] \ \mathrm{m}$ from the intersection, and the crossing order was computed with the heuristic of \cite{Hult2018b_kappa}.
As  \fig{\ref{fig:pdip_approximateRECArandomEval}} demonstrates, the sub-optimality induced by the parameterized constraints is below $0.1 \%$ in more than $50\%$ of the cases.
The small impact is illustrated in \fig{\ref{fig:pdip_recaApproximationExample}}, which shows the difference in the optimal velocity profiles for the scenario corresponding  to the median sub-optimality. 
\iffigs
\begin{figure}[t]
	\centering
	\includegraphics[width=1\linewidth]{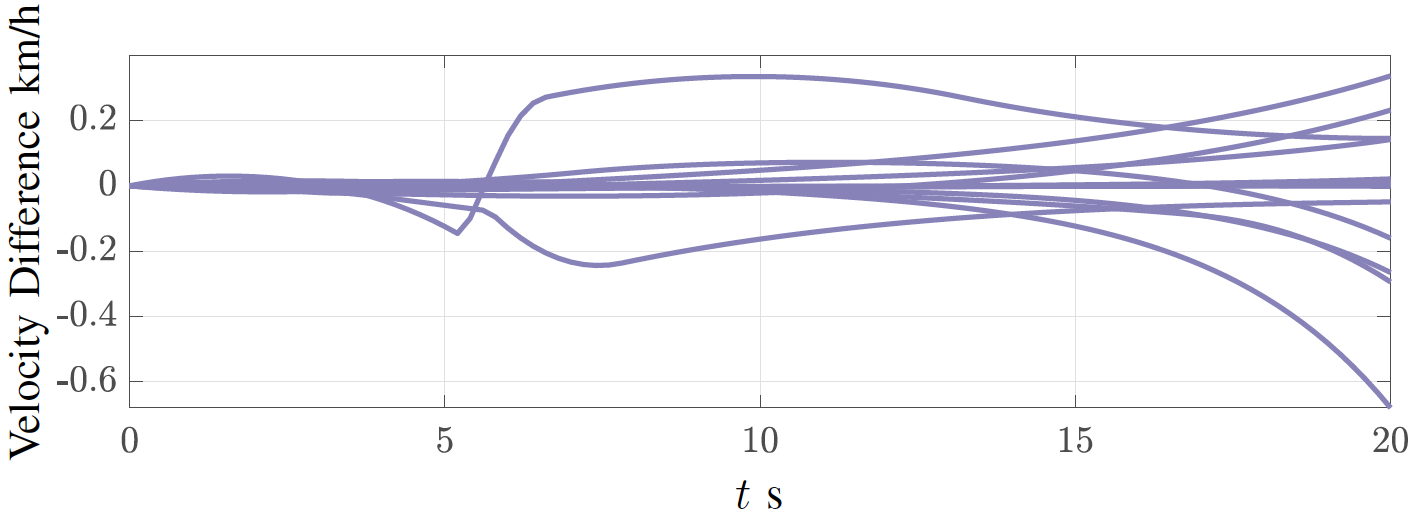}
	\caption{Difference between the optimal velocity profiles and those obtained using the parameterized RECA constraints \eqref{eq:pdip_parametrizedRECA} for a 16 vehicle scenario.}\label{fig:pdip_recaApproximationExample}
\end{figure}
\fi
Interestingly, the difference between the optimal control commands at $k=0$ in the two solutions is smaller than $0.013 \%$ of the input range for all vehicles.
This is below the quantization error of many actuators, such that the difference might not be noticed in practice.
Finally,  more ``flexible'' parameterizations of $\rho_{i}$ could be used to reduce sub-optimality, e.g., by including additional linear segments or by using higher-order polynomials.

\subsection{Reduction of the number of communication rounds}\label{sec:pdip_communicationRoundReduction}
We ought to stress that we intentionally selected a simple implementation of a primal-dual interior-point algorithm when deriving Algorithm~\ref{alg:ap_distributedIP}. This choice has been done both for simplicity of exposition and in order to focus on the main contributions of this paper. Note, however, that more refined update rules for the barrier parameter $\tau$ and Predictor-Corrector strategies~\cite{Nocedal2006_kappa} could be employed to improve convergence.

A simple way to reduce the number of iterations consists in solving the problem to a rough accuracy. 
As remarked in the discussion on \fig{\ref{fig:pdip_exampleVelocity}}, practically acceptable solutions can be obtained for  $\tau$ larger than relevant tolerances on $\|r(z)\|$, since this entails accurately satisfying the constraints while accepting some degree of suboptimality. 
\iffigs
\begin{figure}[t]
	\centering
	\includegraphics[width=1\linewidth]{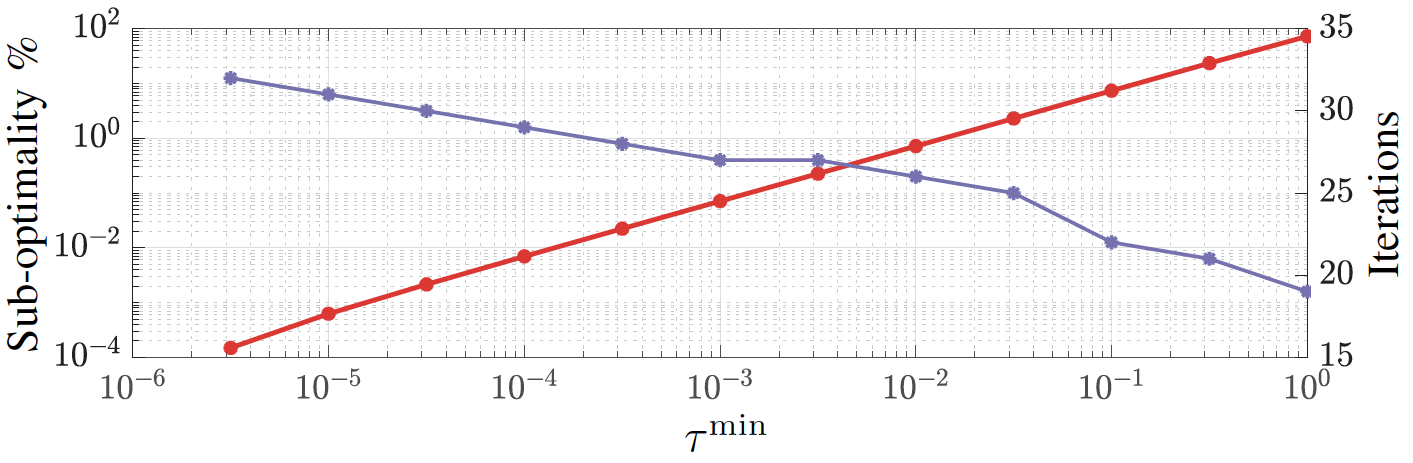}
	\caption{Sub-optimality (red)  and number of iterations (blue) required to reach $\|r(z)\|_\infty\leq 10^{-6} $ for different value of $\tau^\mathrm{min}$.}\label{fig:pdip_tauSubopt}
\end{figure}
\begin{figure}[t]
	\centering
	\includegraphics[width=1\linewidth]{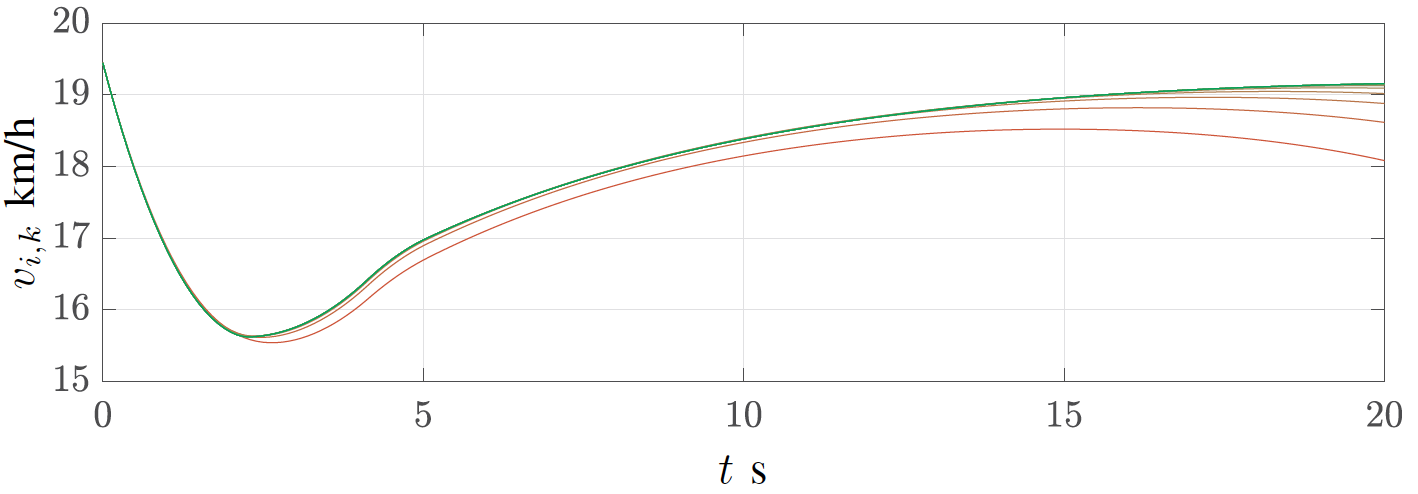}
	\caption{Velocity profiles for $\|r(z)\|_\infty\leq 10^{-6}$: hues between red ($\tau^\mathrm{min}=1$) and green ($\tau^\mathrm{min}=10^{-6}$). }\label{fig:pdip_tauVelocity}
\end{figure}
\fi
In \fig{\ref{fig:pdip_tauSubopt}} we display the results from the scenario considered in Section~\ref{sec:ap_example}, where  $\tau$ is prevented from being smaller than $\tau^\mathrm{min}$, for $\tau^\mathrm{min}$ between $1$ and $10^{-6}$, with $\varepsilon=10^{-6}$ for all cases.
The sub-optimality induced and the number iterations required to reach $\|r(z)\|\leq \varepsilon$ is shown in \fig{\ref{fig:pdip_tauSubopt}}.
Note for instance that $23$ iterations are required for $\tau^\mathrm{min}=10^{-2}$, compared to $33$ in case of $\tau^\mathrm{min}=10^{-6}$, at the expense of less than $1 \%$ sub-optimality. 
The optimal velocity profiles  for one vehicle  for the different values of $\tau^\mathrm{min}$ is  shown in \fig{\ref{fig:pdip_tauVelocity}} (c.f. \fig{\ref{fig:pdip_exampleVelocity}}).
The difference with respect to  the optimal solution is small enough to be practically irrelevant for all but the highest value of  $\tau^\mathrm{min}$. 

\section{Conclusions}\label{sec:ap_discussion}
In this paper we presented tailored linear algebra for Primal-Dual Interior-Point algorithms to be deployed for the optimal coordination of automated vehicles at intersections under a fixed crossing order.
The algorithm inherits the strong convergence guarantees of centralized PDIP solvers and makes it possible to efficiently include rear-end collision avoidance constraints.
We showed that the problem is structured so that the KKT-system can be solved in a hierarchical way, where most operations are parallelized and solved separately for all vehicles and for all lanes.
Additionally, the step size selection through backtracking on a merit function can be distributed under the same pattern.
To reduce the data exchange, we proposed a parameterized and slightly conservative reformulation of the rear-end collision avoidance constraints, and demonstrated its effectiveness in simulations.

We are currently investigating formulations of the coordination problem that allows removal of the restrictive assumption of full CAV penetration. 
We also aim at extending our approach  to scenarios with several connected intersections. Finally, future work will consider testing the algorithm in challenging experimental scenarios.

\ifCLASSOPTIONcaptionsoff
  \newpage
\fi

\bibliographystyle{IEEEtran}
  \bibliography{bibliography}
  
  \begin{IEEEbiography}[{\includegraphics[width=1in,height=1.25in,clip,keepaspectratio]{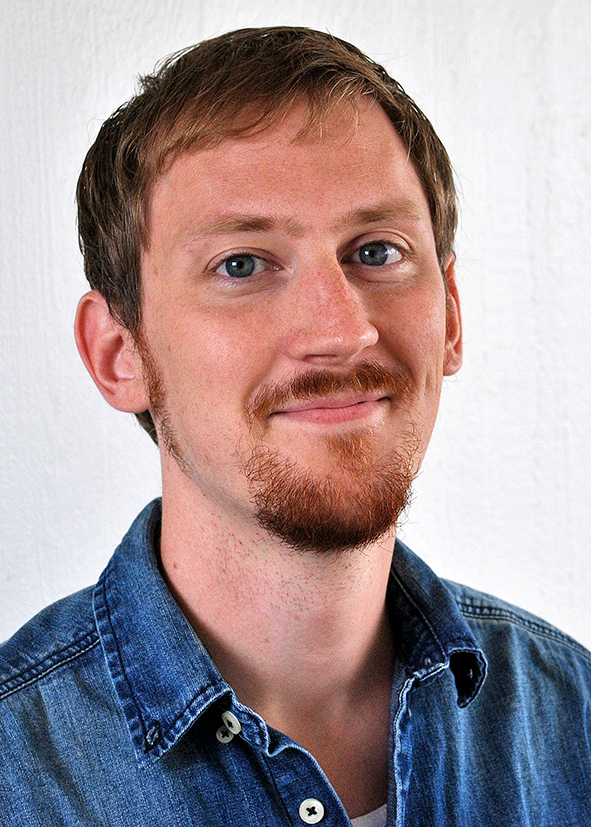}}]{Robert Hult} received the Masters degree in Systems, Control and Mechatronics in 2013, and the Ph.D. degree in 2019, both from Chalmers University of Technology, Sweden. He is currently a research engineer at Volvo Autonomous Solutions in Gothenburg Sweden, working with automated trucks and construction machines. His research interests include distributed and cooperative predictive control, in particular with applications to cooperative vehicles and intelligent transportation systems and site planning for automated vehicles in confined~areas.
  \end{IEEEbiography}
  \vspace{-1em}
  \begin{IEEEbiography}[{\includegraphics[width=1in,height=1.25in,clip,keepaspectratio]{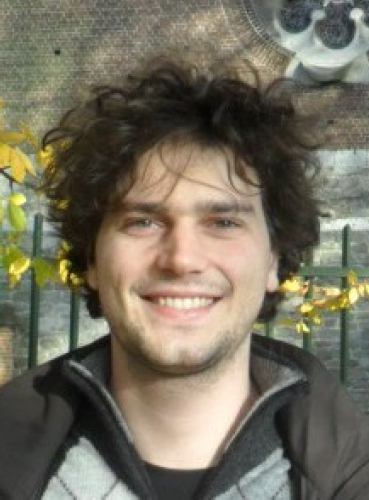}}]{Mario Zanon} received the Master's degree in Mechatronics from the University of Trento, and the Dipl\^{o}me d'Ing\'{e}nieur from the Ecole Centrale Paris, in 2010. After research stays at the KU Leuven, University of Bayreuth, Chalmers University, and the University of Freiburg he received the Ph.D. degree in Electrical Engineering from the KU Leuven in November 2015. He held a Post-Doc researcher position at Chalmers University until the end of 2017. From 2018 until 2021 he has been Assistant Professor and from 2021 he is Associate Professor at the IMT School for Advanced Studies Lucca. His research interests include numerical methods for optimization, economic MPC, reinforcement learning, optimal control and estimation of nonlinear dynamic systems, in particular for aerospace and automotive applications.
  \end{IEEEbiography}
  \vspace{-1em}
  \begin{IEEEbiography}[{\includegraphics[width=1in,height=1.25in,keepaspectratio]{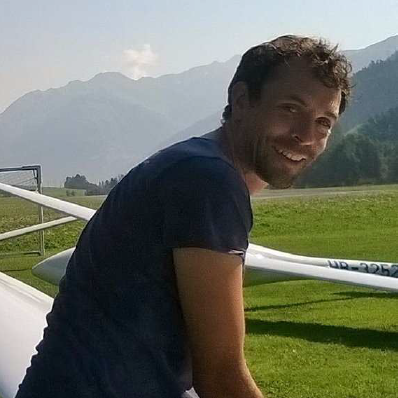}}]{S\'ebastien Gros}
  	received his Ph.D degree from EPFL, Switzerland, in 2007. After a journey by bicycle from Switzerland to the Everest base camp in full autonomy, he joined a R\&D group hosted at Strathclyde University focusing on wind turbine control. In 2011, he joined the university of KU Leuven, where his main research focus was on optimal control and fast MPC for complex mechanical systems. He joined the Department of Signals and Systems at Chalmers University of Technology, G\"{o}teborg in 2013, where he became associate Prof. in 2017. He is now full Prof. at NTNU, Norway and guest Prof. at Chalmers. His main research interests include numerical methods, real-time optimal control, reinforcement learning, and the optimal control of energy-related applications.
  \end{IEEEbiography}
  \vspace{-1em}
  \begin{IEEEbiography}[{\includegraphics[width=1in,height=1.25in,clip,keepaspectratio]{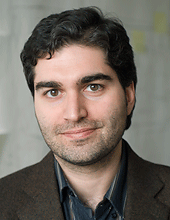}}]{Paolo Falcone}
  	received his Ph.D. degree in Information Technology in 2007 from the University of Sannio, in Benevento, Italy. He is Associate Professor at the Department of Engineering of the University of Modena and Reggio Emilia, Modena, Italy and Professor at the Department of Electrical Engineering at Chalmers University of Technology, Gothenburg, Sweden. His research focuses on constrained optimal control applied to autonomous and semi-autonomous mobile systems, cooperative driving and intelligent vehicles. He is involved in several projects, in cooperation with industry, focusing on autonomous driving, cooperative driving and vehicle dynamics control.
  \end{IEEEbiography}

\end{document}